%% file: main.tex
\documentclass[acmsmall,dvipsnames]{acmart}

\input{header}

\title{Polynomial Bounds of CFLOBDDs against BDDs}

\author{Xusheng Zhi}
\affiliation{%
\institution{University of Wisconsin-Madison}
\country{USA}
}
\affiliation{%
\institution{Peking University}
\country{China}
}
\email{xusheng142@gmail.com}

\author{Thomas Reps}
\orcid{0000-0002-5676-9949}
\affiliation{%
\institution{University of Wisconsin-Madison}
\country{USA}
}
\email{reps@cs.wisc.edu}

\thanks{Some definitions and explanatory material in \Cref{Se:Backgrounds} and \Cref{Appendix: StructuralInvariants} are taken from Sistla et al.~\cite{TOPLAS:SCR24}.}

\input{0_abstract}



\begin{document}

\maketitle

\input{1_intro}
\input{2_backgrounds}

\input{3_overview}
\input{4_relations}
\input{5_counting}
\input{6_tight-instances}
\input{7_related}
\input{8_conclusion}

\begin{acks}
This work was supported, in part,
by a gift from
\grantsponsor{00001}{Rajiv and Ritu Batra}{}, and by
\grantsponsor{00003}{NSF}{https://www.nsf.gov/}
under grants
\grantnum{00003}{CCF-2211968}
and
\grantnum{00003}{CCF-2212558}.
Any opinions, findings, and conclusions or recommendations
expressed in this publication are those of the authors,
and do not necessarily reflect the views of the sponsoring
entities.
Thomas Reps has an ownership interest in GrammaTech, Inc.,
which has licensed elements of the technology reported
in this publication.
\end{acks}

\bibliographystyle{ACM-Reference-Format}
\bibliography{CFLOBDD}


\input{9_appendix-defCFLOBDD}

\input{10_appendix-proofLemCorr}

\end{document}

%% file: header.tex
\usepackage{graphicx} 
\usepackage{xcolor} 
\usepackage{amsfonts}
\usepackage{amsmath}
\usepackage{amsthm}
\usepackage{mathtools}
\usepackage{graphicx}
\usepackage{soul}

\usepackage{hyperref}
\usepackage{cleveref}
\crefformat{section}{#2\S#1#3}
\Crefformat{section}{#2\S#1#3}

\usepackage[framemethod=default]{mdframed}
\mdfsetup{skipabove=.67\topskip,skipbelow=.67\topskip}
\usepackage{enumitem}
\usepackage[titletoc]{appendix}

\newcommand{\eqdef}{\buildrel \mbox{\tiny\rm def} \over =}
\newcommand{\depth}{{\textit{depth}}}
\newcommand{\NS}{{\textit{NS}}}

\newtheorem{formula}{Formula}[section]

\usepackage{algorithm2e}
\SetKwInput{Input}{Input}
\SetKwInput{Output}{Output}
\SetKwBlock{Begin}{begin}{end}
\RestyleAlgo{ruled}

\usepackage{enumitem}
\usepackage{hhline}
\usepackage{multirow}

\newtheorem{MyTheorem}{Theorem}[section]
\newtheorem{Proposition}[MyTheorem]{Proposition}

{\par\noindent{\bf Proposition~\ref{Prop:#1}\/}\begin{em}}%
{\end{em}}
\newtheorem{Lemma}[MyTheorem]{Lemma}

\newenvironment{SLem}[1]%
{\par\noindent{\bf Lemma~\ref{Lem:#1}\/}\begin{em}}%
{\end{em}}

\newcommand{\MPP}{\mbox{Matched-Path Principle\/}}
\newcommand{\CIP}{\mbox{Contextual-Interpretation Principle\/}}

\newcommand{\Matched}{{\textit{X}}}

\newcommand{\ADD}{{\textit{ADD}}}
\newcommand{\wasted}{{\textit{wasted}}}

\newcommand {\revision}[1]{{{#1}}}


%% file: 0_abstract.tex
\begin{abstract}
Binary Decision Diagrams (BDDs) are widely used for the representation of Boolean functions.
Context-Free-Language Ordered Decision Diagrams (CFLOBDDs) are a plug-compatible replacement for BDDs---roughly, they are BDDs augmented with a certain form of procedure call.
A natural question to ask is, ``For a given
\revision{
family of Boolean functions $F$, what is the relationship between the size of a BDD for $f \in F$ and the size of a CFLOBDD for $f$?''
\citeauthor{TOPLAS:SCR24} established that there are \emph{best-case families of functions}, which demonstrate an inherently exponential separation between CFLOBDDs and BDDs.
They showed that there are families of functions $\{ f_n \}$ for which, for all $n = 2^k$, the CFLOBDD for $f_n$ (using a particular variable order) is exponentially more succinct than \emph{any} BDD for $f_n$ (i.e., using any variable order).
However,
}
they did not give a \emph{worst-case bound}---i.e., they left open the question, ``Is there a family of functions \revision{$\{ g_i \}$} for which the size of a CFLOBDD for \revision{$g_i$} must be substantially larger than a BDD for \revision{$g_i$?''} For instance, it could be that there is a family of functions for which the BDDs are exponentially more succinct than any corresponding CFLOBDDs.

This paper studies such questions, and answers the second question posed above in the negative.
In particular, we show that by using the same variable ordering in the CFLOBDD that is used in the BDD, the size of a CFLOBDD for any function \revision{$h$} cannot be far worse than the size of the BDD for \revision{$h$.} The bound that relates their sizes is polynomial:
\begin{quote}
      If BDD $B$ for function \revision{$h$} is of size $|B|$ and uses variable ordering $\textit{Ord}$, then the size of the CFLOBDD $C$ for \revision{$h$} that also uses $\textit{Ord}$ is bounded by $\revision{\mathcal{O}}(|B|^3)$.
\end{quote}
The paper also shows that the bound is tight:
there is a family of functions for which $|C|$ grows as $\Omega(|B|^3)$.
\end{abstract}

%% file: 1_intro.tex
\section{Introduction}
\label{Se:Introduction}


Binary Decision Diagrams (BDDs)
\revision{
\cite{toc:Bryant86,dac:BRB90,Book:Wegener00}
}
are commonly used to represent Boolean functions, offering a compressed representation of decision trees. Context-Free-Language Ordered Binary Decision Diagrams (CFLOBDDs)
\revision{
\cite{TOPLAS:SCR24}
}
are essentially a plug-compatible alternative to BDDs, but based on a different function-decomposition principle. Whereas a BDD can be considered to be a special form of bounded-size, branching, but non-looping program, a CFLOBDD can be considered to be a bounded-size, branching, but non-looping program in which a certain form of \emph{procedure call} is permitted. CFLOBDDs share many good properties of BDDs, but---in the best case---the CFLOBDD for a Boolean function can have a double-exponential reduction in size compared to the corresponding decision tree.\footnote{
\revision{
    An implementation of CFLOBDDs is available on GitHub \cite{SOFTWARE:CFLOBDDs}.
}
}

It is natural to ask how the sizes of the BDD and CFLOBDD for a given function compare. Sistla et al.~\cite[\S8]{TOPLAS:SCR24} established that in the best case, the CFLOBDD for a function $f$ can be exponentially smaller than any BDD for $f$ (regardless of what variable ordering is used in the BDD);
however, they do not give a bound in the opposite direction
(i.e., a bound on CFLOBDD size as a function of BDD size, for all BDDs).
They spell out two possibilities as follows \cite[\S8]{TOPLAS:SCR24}:
\begin{quote}
  It could be that there are families of functions for which BDDs are exponentially more succinct than any corresponding CFLOBDD;
  however, it could also be that for every BDD there is a corresponding CFLOBDD no more than, say, a polynomial factor larger.
\end{quote}
In this paper, somewhat surprisingly, we establish that for every BDD, the size of the corresponding CFLOBDD is at most a polynomial function of the BDD's size:
\begin{quote}
  If BDD\footnote{
  \revision{More precisely, the stated bound holds when $B$ is a quasi-reduced BDD (see \Cref{Se:BDDs}).
  The size of a quasi-reduced BDD is at most a factor of $n + 1$ larger than the size of the corresponding BDD \cite[Thm.\ 3.2.3]{Book:Wegener00}, where $n$ is the number of Boolean variables.
  For $B$ a true BDD, the bound becomes $\revision{\mathcal{O}}(n^3 |B|^3)$.
  For convenience, this paper focuses on quasi-reduced BDDs.
  }
  }
  $B$ for function $f$ is of size $|B|$ and uses variable ordering $\textit{Ord}$, then the size of the CFLOBDD $C$ for $f$ that also uses $\textit{Ord}$ is bounded by $\revision{\mathcal{O}}(|B|^3)$.
\end{quote}
Moreover, in \Cref{Se:TightInstances} we show that this bound is tight (i.e., asymptotically optimal) by constructing a family of functions for which $|C|$ grows as $\Omega(|B|^3)$.

\revision{
As will be explained in \Cref{Se:RelatedWork}, CFLOBDDs make use of a tree-structured decomposition hierarchy on the Boolean variables.
There are other kinds of decision diagrams that also use a tree-structured decomposition of the Boolean variables \cite{SDD,VSSDD}.
In those structures, the variable decomposition is a user-supplied input, and in the observations that have been made about the relative sizes of those structures compared to BDD sizes, only the \emph{best} variable-decomposition tree is considered.
However, the size-comparison question is then trivial, because
selecting a linear tree---i.e., a \emph{linked list} of the Boolean variables---trivially emulates BDDs:
it produces a data structure whose size is linearly related to that of the corresponding BDD.

In contrast, our work establishes a \emph{worst-case bound}:
we show that a CFLOBDD, using a \emph{binary-tree-structured decomposition} of the variables, is never more than cubic in the size of the corresponding BDD.
The different nature of our result is important because the tree-structured variable decomposition is what enables CFLOBDDs---and other decision diagrams \cite{SDD,VSSDD}---to attain a best-case exponential-succinctness advantage over BDDs.

\noindent
\begin{mdframed}
\revision{
  \textit{This paper is the first to demonstrate that a data structure can enjoy an exponential-succinctness advantage over BDDs in the best case, while simultaneously being at most polynomially larger than BDDs in the worst case.}
}
\end{mdframed}

\noindent
\Cref{Se:RelatedWork} gives a more detailed discussion of the issue of data-structure sizes compared to BDD sizes, for CFLOBDDs and other decision-diagram data structures that use a tree-structured decomposition of the Boolean variables.
}

\revision{
\citet{DBLP:journals/jair/DarwicheM02} initiated a research program on the properties of decision diagrams, exploring trade-offs in the relative size and computational complexity of different families of data structures.
Their work showed that many known data structures for propositional knowledge bases are obtained by imposing specific restrictions on Negation Normal Form (NNF).
According to their definition,
\begin{quote}
  A sentence in $\textrm{NNF}_\textit{PS}$ is a rooted, directed-acyclic graph (DAG) where each leaf node is labeled with $\textit{true}$, $\textit{false}$, $X$, or $\neg X$, [for] $X \in \textit{PS}$;
  and each internal node is labeled with $\land$ or $\lor$ and can have arbitrarily many children,
\end{quote}
where $\textit{PS}$ is a denumerable set of propositional variables  \cite[Definition 2.1]{DBLP:journals/jair/DarwicheM02}.
As will become clear in \Cref{Se:CFLOBDDs}, CFLOBDDs are \emph{not} DAGs, and hence lie outside their framework.
This paper establishes a result that compares our non-NNF structure to BDDs, which are NNF structures and do lie within their framework.
Thus, our results are complementary to the work of \citeauthor{DBLP:journals/jair/DarwicheM02}.
}

\paragraph{Organization.}
\Cref{Se:Backgrounds} reviews the basic definitions of BDDs and CFLOBDDs,
\revision{
and introduces several items of terminology.
\Cref{Se:Overview} presents the issue of ``3/4-depth duplication'' which appears to be an obstacle to the existence of the desired polynomial bound;
\Cref{Se:Overview} then provides intuition about why a polynomial bound does indeed exist---which guides the establishment of the bound in subsequent sections.}
\Cref{Se:Relations} establishes some structural relationships between BDDs and CFLOBDDs. 
\Cref{Se:Counting} presents upper bounds on the number of groupings, vertices, and edges in a CFLOBDD as a function of the size of the corresponding BDD.
\Cref{Se:TightInstances} gives tight instances for the three bounds.
\revision{
\Cref{Se:RelatedWork} discusses related work.
}
\Cref{Se:Conclusion} concludes. 

%% file: 2_backgrounds.tex
\section{Background and Terminology}
\label{Se:Backgrounds}

This section reviews BDDs (\Cref{Se:BDDs}) and CFLOBDDs (\Cref{Se:CFLOBDDs}), and introduces some terminology conventions used in the paper (\Cref{Se:TerminologyConventions}).

\subsection{Binary Decision Diagrams}
\label{Se:BDDs}

Binary Decision Diagrams (BDDs)
\revision{
\cite{toc:Bryant86,dac:BRB90,Book:Wegener00}
}
are commonly used to represent Boolean-valued and non-Boolean-valued functions over Boolean arguments (i.e., $\{0,1\}^n \rightarrow \{0,1\}$ and $\{0,1\}^n \rightarrow V$, respectively, for some value domain $V$).
BDDs are directed acyclic graphs (DAGs) in which maximal sub-DAGs are shared.
An Ordered BDD (OBDD) is a BDD in which the same variable ordering is imposed on the Boolean variables, and can thus be thought of as a compressed decision tree.
In an OBDD, a node $m_1$ can have the same node $m_2$ as its false-successor and its true-successor, in which case $m_1$ is called a \emph{don't-care node}.
In a Reduced Ordered BDD (ROBDD), all don't-care nodes are removed by repeatedly applying a ply-skipping transformation so that for each node $m$, its false-successor and true-successor are different nodes.
An OBDD in which don't-care nodes are \emph{not} removed (i.e., plies are never skipped) is sometimes called a \emph{quasi-reduced OBDD} \cite[p.\ 51]{Book:Wegener00}.\footnote{
\revision{
  For introductory material about BDDs, the reader should consult Brace et al.\ \cite{dac:BRB90} or the tutorial by Andersen \cite{Notes:Andersen98}.
  Wegener\ \cite{Book:Wegener00} provides a thorough account of the properties of BDDs and discusses many BDD variants.
}
}
\revision{Consequently, in a quasi-reduced BDD all paths from the root to a terminal value have length $n$, where $n$ is the number of variables.}
An ROBDD with non-binary-valued terminals is called a Multi-Terminal BDD (MTBDD) \cite{dac:CMZFY93,CMU:CS-95-160} or an Algebraic Decision Diagram (ADD) \cite{iccad:BFGHMPS93}.
A quasi-reduced-OBDD with non-binary-valued terminals could be called a quasi-reduced MTBDD.

The interpretation of a \revision{quasi-reduced MT}BDD with respect to a given assignment to the Boolean variables is the same as the process of determining whether a given string is accepted by a deterministic finite-state automaton. We start from the root node; for each successive variable, we choose a successor of the current node based on whether the variable's value in the assignment is 0 or 1. After interpreting all the variables, we arrive at a leaf node (also known as a value node or terminal node). The value in that leaf node is the result of the interpretation.

\begin{figure}
    \centering
    \includegraphics[width=0.3\linewidth]{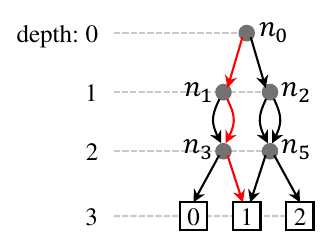}
    \caption{
\revision{
    (Color online.)
}
    The BDD for $f_0 = \lambda b_0, b_1, b_2 \,.\,(\text{if } b_0 \text{ then }1 \text{ else }0) + (\text{if } b_2 \text{ then }1 \text{ else }0)$. }
    \label{fig:BDDexample}
\end{figure}

\begin{example}\label{Exa:BDD}
    Let $f_0(b_0, b_1, b_2) \eqdef (\text{if } b_0 \text{ then }1 \text{ else }0) + (\text{if } b_2 \text{ then }1 \text{ else }0)$.
    The \revision{quasi-reduced MT}BDD for $f_0$ is shown in \Cref{fig:BDDexample}.
    Following the conventions used in \cite{toc:Bryant86}, we depict the non-terminal nodes as circles and the leaf nodes as squares.
    Throughout the paper, the left branch is always the transition taken for 0, and the right branch is always the transition taken for 1.
    The path taken when interpreting $[b_0 \mapsto 0, b_1 \mapsto 1, b_2 \mapsto 1]$ is depicted in red.
    We obtain the value 1, which matches the result we get by putting $[b_0 \mapsto 0, b_1 \mapsto 1, b_2 \mapsto 1]$ into the definition of $f_0$ and simplifying.
\end{example}

If $B$ is a (Boolean or Multi-Terminal) OBDD/ROBDD, there are two natural quantities for expressing the cost of a problem that involves $B$, namely,
\begin{itemize}
  \item
    the \emph{size} of $B$, denoted by $|B|$, which is the number of nodes (or edges) in $B$, and
  \item
    $n$, the number of Boolean variables over which $B$ is interpreted.
\end{itemize}
The size of a quasi-reduced OBDD is at most a factor of $n+1$ larger than the size of the corresponding ROBDD \cite[Thm.\ 3.2.3]{Book:Wegener00}.

In this paper, we establish bounds on the size of a CFLOBDD $C$ in terms of the size of the corresponding quasi-reduced OBDD $B$.
Thus, if the bound is $\revision{\mathcal{O}}(h(|B|))$, the bound with respect to ROBDDs is $\revision{\mathcal{O}}(h(n|B|))$---e.g., an $\revision{\mathcal{O}}(|B|^3)$ bound with respect to quasi-reduced OBDDs becomes an $\revision{\mathcal{O}}(n^3 |B|^3)$ bound with respect to ROBDDs.
\revision{
We will refer to Boolean-valued and non-Boolean-valued (``pseudo-Boolean'') functions over Boolean arguments generically as Boolean functions, and
}
refer to quasi-reduced OBDDS and quasi-reduced MTBDDs generically as BDDs from hereon.

The depth of a BDD node is defined as the distance from the root.\footnote{
  In some other papers, the ``depth'' of a node is called the ``ply'' at which the node occurs in the BDD.
}
For example, in \Cref{fig:BDDexample} the depths of the nodes are indicated on the left.
It is clear that the depth of a node equals the number of variables that have been interpreted when the node is reached.

Given a function $f : \{0,1\}^n \rightarrow V$ and a partial assignment $p$ that gives values for the first $d$ variables, the \emph{residual function} of $f$ with respect to $p$ is the function of type $\{0,1\}^{n-d} \rightarrow V$ obtained by currying $f$ with respect to its first $d$ parameters and evaluating the curried $f$ on $p$ \cite[\S4.2]{Book:JGS93}.
An important property of BDDs is that each BDD node represents a residual function, and each different node represents a different residual function.
For a BDD with $n$ variables, a node at depth $d$ corresponds to a residual function with $n - d$ variables.
For the example in \Cref{fig:BDDexample},
node $n_1$ represents ``$\lambda b_1, b_2 \,.\,$(if $b_2$ then 1 else 0)''; node $n_5$ represents ``$\lambda b_2 \,.\,$(if $b_2$ then 2 else 1)''; each leaf node represents a constant value, which can be viewed as a function of zero Boolean variables.

\subsection{Context-Free-Language Ordered Binary Decision Diagrams}
\label{Se:CFLOBDDs}

Context-Free-Language Ordered Binary Decision Diagrams (CFLOBDDs) \revision{\cite{TOPLAS:SCR24}} are decision diagrams that employ a certain form of procedure call.
\revision{
As with BDDs, they can be used to represent Boolean-valued and non-Boolean-valued functions over Boolean arguments (i.e., $\{0,1\}^n \rightarrow \{0,1\}$ and $\{0,1\}^n \rightarrow V$, respectively, for some value domain $V$).
}
Whereas a BDD can be thought of as an acyclic finite-state machine (modulo ply-skipping in the case of ROBDDs), a CFLOBDD is a particular kind of \emph{single-entry, multi-exit, non-recursive, hierarchical finite-state machine} (HFSM) \cite{TOPLAS:ABEGRY05}.
We will introduce CFLOBDDs using as an example the function
$f_1(b_0, b_1) \eqdef (\text{if (} b_0 \land \neg b_1\text{) then 7 else 5})$.
The CFLOBDD for $f_1$ is shown in \Cref{fig:CFLOBDDexample}. 

\begin{figure}
    \centering
    \includegraphics[width=0.7\linewidth]{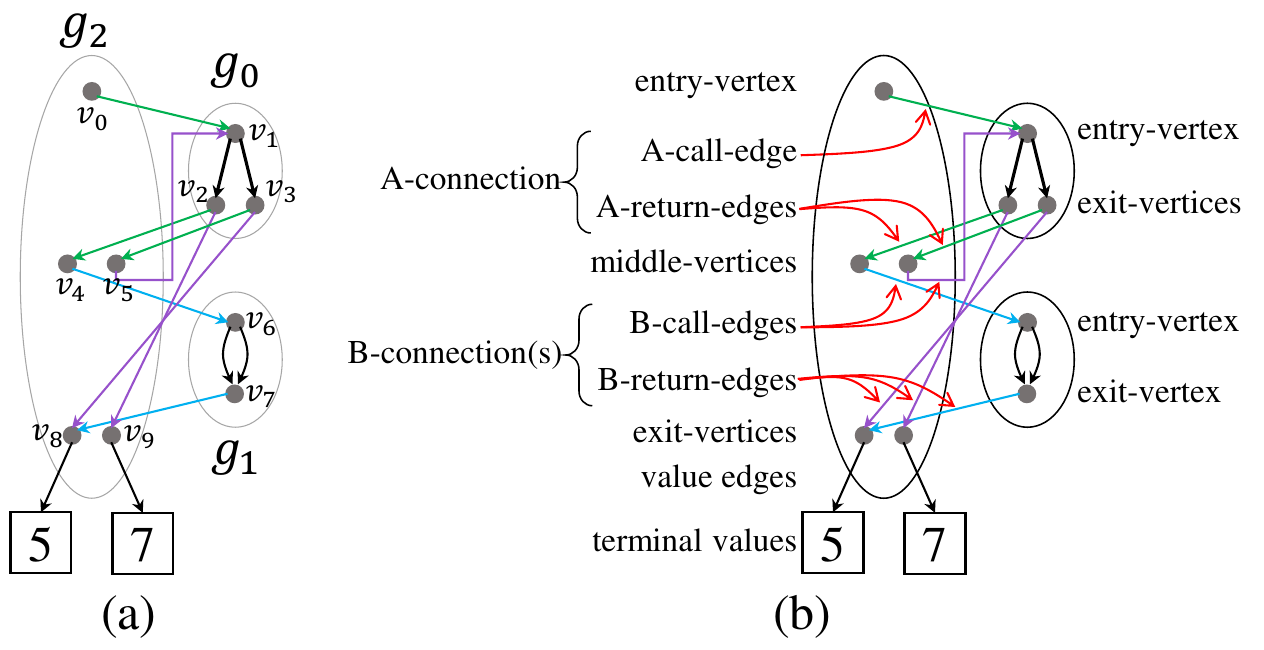}
    \caption{
\revision{
    (Color online.)
}
    The CFLOBDD for $f_1 \eqdef \lambda b_0, b_1 \,.\, (\text{if (} b_0 \land \neg b_1\text{) then 7 else 5})$.
    }
    \label{fig:CFLOBDDexample}
\end{figure}

The formal definition of CFLOBDDs contains two parts:
\revision{
\Cref{Def:MockCFLOBDD} specifies the basic structure of CFLOBDDs (but technically defines a data structure that also includes objects that are not CFLOBDDs);
}
\Cref{De:CFLOBDD} imposes some additional structural invariants.
Where necessary, we distinguish between \emph{mock-CFLOBDDs} (\Cref{Def:MockCFLOBDD}) and \emph{CFLOBDDs} (\Cref{De:CFLOBDD}), although we typically drop the qualifier ``mock-'' when there is little danger of confusion.
\Cref{fig:CFLOBDDexample}(b) illustrates \Cref{Def:MockCFLOBDD} for the CFLOBDD that represents $f_1$.
\revision{
\Cref{fig:CFLOBDDLikeTerms} provides a guide to the modifiers that we use on the term ``CFLOBDD'' in \Cref{Def:MockCFLOBDD}, \Cref{De:MockProtoCFLOBDD}, and \Cref{De:CFLOBDD}.
}

\begin{figure}
  \centering
  \small{
\revision{
  \begin{tabular}{||r|c||p{0.375\linewidth}|p{0.375\linewidth}||}
    \hhline{~~|t:==:t|}
    \multicolumn{2}{c||}{}     & \multicolumn{2}{c||}{``proto-''} \\
    \hhline{||~~--}
    \multicolumn{2}{c||}{}     & \multicolumn{1}{c|}{with}  &  \multicolumn{1}{c||}{w/o} \\
    \hhline{|t:==#=|=||}
    \multirow{2}{*}{``mock-'' \rule{0em}{14.4ex}} & \multirow{1}{*}{with \rule{0em}{8.0ex}} & A ``procedure-like'' substructure of a hierarchical decision diagram.  Does not necessarily conform to the structural constraints of \Cref{De:CFLOBDD}.   & A hierarchical decision diagram with ``procedure calls'' that (i) represents a Boolean function, but (ii) does not necessarily conform to the structural constraints of \Cref{De:CFLOBDD}. \\
    \hhline{||~---||}
                              & \multirow{1}{*}{w/o \rule{0em}{8.0ex}}  & The ``procedure-like'' sub-structures over which inductive arguments about CFLOBDDs are made.  They conform to the structural constraints of \Cref{De:CFLOBDD}. & Hierarchical decision diagrams that provide a canonical representation of each Boolean function.  They conform to the structural constraints of \Cref{De:CFLOBDD}. \\ 
    \hhline{|b:====:b|}
  \end{tabular}
}
  }
  \caption{\revision{Guide to the modifiers on the term ``CFLOBDD(s).''}
  }
  \label{fig:CFLOBDDLikeTerms}
\end{figure}

\begin{definition}[Mock-CFLOBDD]\label{Def:MockCFLOBDD}
    A \emph{mock-CFLOBDD} at level $k$ is a hierarchical structure made up of some number of \emph{groupings} (depicted as large ovals in the diagrams throughout the paper), of which there is one grouping at level $k$, and at least one at each level $0, 1, \ldots, k-1$.
    The grouping at level $k$ is the \emph{head} of the mock-CFLOBDD.
\revision{
    A $k$-level mock-CFLOBDD represents a Boolean function (or pseudo-Boolean function) over $2^k$ variables.
}
    
    A grouping is a collection of vertices (depicted as black dots)
    and edges to other groupings (depicted as arrows).

    \begin{itemize}
      \item 
        Each grouping $g$ at level $0 \le l \le k$ has a unique \emph{entry vertex}, which is disjoint from $g$'s non-empty set of \emph{exit vertices}.      
      \item
        If $l = 0$, $g$ is either a \emph{fork grouping} ($g_0$ in \Cref{fig:CFLOBDDexample}(a)) or a \emph{don't-care grouping} ($g_1$ in \Cref{fig:CFLOBDDexample}(a)).
        The entry vertex of a level-$0$ grouping corresponds to a decision point:
        left branches are for $F$ (or $0$);
        right branches are for $T$ (or $1$).
        A don't-care grouping has a single exit vertex, and the edges for the left and right branches both connect the entry vertex to the exit vertex.
        A fork grouping has two exit vertices:
        the entry vertex's left and right branches connect the entry vertex to the first and second exit vertices, respectively.
      \item
        If $l \ge 1$, $g$ has a further disjoint set of \emph{middle vertices}.
        We assume that both the middle vertices and the exit vertices are associated with some fixed, known total order (i.e., the sets of middle vertices and exit vertices could each be stored in an array).
        Moreover, $g$ has an \emph{A-call edge} that, from $g$'s entry vertex, ``calls'' a level-$(l\text{-}1)$ grouping $a$, along with a set of matching \emph{\revision{A-}return edges};
        each return edge from $a$ connects one of the exit vertices of $a$ to one of the middle vertices of $g$.
        In addition, for each middle vertex $m_j$, $g$ has a \emph{B-call} edge that ``calls'' a level-$(l\text{-}1)$ grouping $b_j$, along with a set of matching \emph{\revision{B-}return edges};
        each return edge from $b_j$ connects one of the exit vertices of $b_j$ to one of the exit vertices of $g$.

        \hspace{1.5ex}
        \revision{
        Note that the ``matching'' relation is one-many between call edges and return edges. In other words, each call edge has a disjoint set of return edges that match it.
        }
        We call $g$'s A-call edge, along with the set of matching A-return edges, an \emph{A-connection};
        we call each of $g$'s B-call edges, along with its respective set of matching B-return edges, a \emph{B-connection}.\footnote{
          The terminology used in this paper differs slightly from that in \cite{TOPLAS:SCR24}.
          Here, the term ``connection'' refers to a call edge and the matching return edges as a whole, whereas ``connection(-edge)'' in \cite{TOPLAS:SCR24}  refers only to the call edge.
        }
        In other words, a grouping $g$ at level $l \ge 1$ has an A-connection that represents a call to (and return from) a level-$(l\text{-}1)$ grouping $a$;
        for each middle vertex $m_j$, $g$ has a B-connection that represents a call to (and return from) a level-$(l\text{-}1)$ grouping $b_j$.
        In diagrams in the paper, we depict the different connections in different colors when there is a need to make a distinction.
        \revision{ For example, there are three connections in the CFLOBDD in \Cref{fig:CFLOBDDexample}: an A-connection that represents a call to $g_0$ (depicted in green), a B-connection that represents a call to $g_1$ (depicted in cyan), and a B-connection that represents another call to $g_0$ (depicted in purple).
        }
      \item
        If $l = k$, $g$ additionally has a set of \emph{value edges} that connect each exit vertex of $g$ to a \emph{terminal value}.
    \end{itemize}
\end{definition}

A \emph{proto-(mock-)CFLOBDD} is a (mock-)CFLOBDD without the value edges and terminal values. In other words, each grouping in a (mock-)CFLOBDD is the head of a proto-(mock-)CFLOBDD; the proto-(mock-)CFLOBDD of the level-$k$ grouping, together with the set of value edges and terminal values, forms a (mock-)CFLOBDD.
\revision{
The proto-(mock-)CFLOBDDs form an inductive structure, and we will make inductive arguments about CFLOBDDs based on this inductive structure:\footnote{
\revision{
  One cannot argue inductively in terms of CFLOBDDs because its constituents are proto-CFLOBDDs, not full-fledged CFLOBDDs.
  Thus, to prove that some property holds for a CFLOBDD, there will typically be an inductive argument to establish a property of the proto-CFLOBDD headed by the outermost grouping of the CFLOBDD, with---if necessary---an additional argument about the CFLOBDD’s value edges and terminal values.
}
}
}

\revision{
\begin{definition}[Mock-proto-CFLOBDD]\label{De:MockProtoCFLOBDD}
  A \emph{mock-proto-CFLOBDD} at level $i$ is a grouping at level $i$, together with the lower-level groupings to which it is connected (and the connecting edges).
  In other words, a mock-proto-CFLOBDD has the following recursive structure:
  \begin{itemize}
    \item
      a mock-proto-CFLOBDD at level 0 is either a fork grouping or a don't-care grouping
    \item 
      a mock-proto-CFLOBDD at level $i$ is headed by a grouping at level $i$ whose
      \begin{itemize}
        \item A-connection ``calls'' a level-($i$-1) mock-proto-CFLOBDD, and
        \item B-connections ``call'' some number of level-($i$-1) mock-proto-CFLOBDDs.
      \end{itemize}
  \end{itemize}
\end{definition}
}

The interpretation of an assignment in a
\revision{
mock-CFLOBDD
}
is with respect to a path that obeys the following principle:

\smallskip
\begin{mdframed}[innerleftmargin = 3pt, innerrightmargin = 3pt, skipbelow=-0.0em]
  $\textbf{\MPP}$.  {\em When a path follows an edge that returns to level $i$ from level $i-1$, it must follow an edge that matches the closest preceding edge from level $i$ to level $i-1$.\/}
\end{mdframed}

\begin{example}\label{Exa:MatchedPaths}
\revision{
In the CFLOBDD for $f_1$ from \Cref{fig:CFLOBDDexample}(a),
}
the matched path that corresponds to the assignment $[x_0 \mapsto 0, x_1 \mapsto 1]$ is as follows:
\[
  v_0 {\color{ForestGreen}{\xrightarrow{\text{call }g_0}}} v_1 \xrightarrow{x_0 \mapsto 0} v_2 {\color{ForestGreen}{\xrightarrow{\text{ret}}}} v_4 {\color{cyan}{\xrightarrow{\text{call } g_1}}} v_6 \xrightarrow{x_1 \mapsto 1} v_7 {\color{cyan}{\xrightarrow{\text{ret}}}} v_8 \xrightarrow{\text{value-edge}} 5.
\]
The matched path that corresponds to the assignment $[x_0 \mapsto 1, x_1 \mapsto 0]$ is as follows:
\[
  v_0 {\color{ForestGreen}{\xrightarrow{\text{call }g_0}}} v_1 \xrightarrow{x_0 \mapsto 1} v_3 {\color{ForestGreen}{\xrightarrow{\text{ret}}}} v_5 {\color{violet}{\xrightarrow{\text{call } g_0}}} v_1 \xrightarrow{x_1 \mapsto 0} v_2 {\color{violet}{\xrightarrow{\text{ret}}}} v_9 \xrightarrow{\text{value-edge}} 7.
\]
Note that the two paths take different return edges after they reach $v_2$.
In the first path, $g_0$ was called along the green edge ($v_0 {\color{ForestGreen}{\to}} v_1$), so the path must continue along the \revision{(matching)} green edge ($v_2 {\color{ForestGreen}{\to}} v_4$) when returning;
in the second path, the second call on $g_0$ was along the purple edge ($v_5 {\color{violet}{\to}} v_1$), so it must continue along the \revision{(matching)} purple edge ($v_2 {\color{violet}{\to}} v_9$) when returning.
\end{example}

\Cref{Exa:MatchedPaths} also illustrates that the
\revision{
same level-$0$ grouping can
}
interpret \emph{different variables at different places in a matched path}, in accordance with the following principle:

\smallskip
\begin{mdframed}[innerleftmargin = 3pt, innerrightmargin = 3pt, skipbelow=-0.0em]
  $\textbf{\CIP}$.
  {\em A level-$0$ grouping is not associated with a specific Boolean variable.
  Instead, the variable that a level-$0$ grouping refers to is determined by context:
  the $n^{\textit{th}}$ level-$0$ grouping visited along a matched path is used to interpret the $n^{\textit{th}}$ Boolean variable.}
\end{mdframed}

\smallskip
\noindent
For instance, in the second of the two paths
\revision{
in \Cref{Exa:MatchedPaths},
}
$g_0$ is called twice:
the first time to interpret $[x_0 \mapsto 1]$;
the second time to interpret $[x_1 \mapsto 0]$.

We can characterize the matched path used to interpret an assignment from the standpoint of each ``step'' taken along the path, which will be used frequently in the proofs given later in the paper.
Let $C(f)$ be a CFLOBDD with maximum level $k$, and $\alpha$ be an assignment to the Boolean variables $x_0$, $x_1$, $\ldots$, $x_{{2^k}-1}$.

A matched path starts from the entry vertex of the outermost grouping of $C(f)$. Then, when we are in a level-$l$ grouping $g$:
\begin{enumerate}[label=(\alph*)]
  \item
    If we are at the entry vertex of $g$:
    \begin{itemize}
      \item
        If $l = 0$, interpret the next Boolean variable in assignment $\alpha$, and go to one of the exit vertices of $g$.
      \item
        If $l > 0$, go to a level-$(l\text{-}1)$ grouping by following the A-call edge that leaves this entry vertex.
    \end{itemize}
  \item
    If we are at one of the middle vertices, it is guaranteed that $l > 0$ (because level-0 groupings have no middle vertices).
    Go to a level-$(l\text{-}1)$ grouping by following the B-call edge that leaves the current middle vertex.
  \item
    If we are at one of the exit vertices: 
    \begin{itemize}
      \item
        If $l = k$, follow the value edge, and return the terminal value as the result of the interpretation.
      \item
        If $l < k$, return to a vertex of a level-$(l\text{+}1)$ grouping, following the return edge that \emph{matches} the call edge used to enter $g$.
    \end{itemize}
\end{enumerate}

\noindent
The length of every matched path through a level-$l$ proto-CFLOBDD is described by the following recurrence equation:
\begin{equation}
  \label{Eq:MatchedPathLengthRecurrence}
  L(l) = 2 L(l-1) + 4  \qquad\qquad  L(0) = 1,
\end{equation}
which has the solution $L(l) = 5 \times 2^l -4$.
Consequently, each matched path through a CFLOBDD is finite, and the interpretation process must always terminate.

\revision{
Let $\textit{Var}(i)$ be the number of variables interpreted along each matched path in a level-$i$ proto-CFLOBDD (i.e., the number of times $\textit{Var}(i)$ that each matched path in a level-$i$ proto-CFLOBDD reaches the entry vertex of a level-$0$ grouping). In particular, $\textit{Var}(i)$ is described by the following recurrence relation:
\begin{equation}
  \label{Eq:MatchedPathVariableInterpretationRecurrence}
  \textit{Var}(l) = 2 \textit{Var}(l-1)  \qquad\qquad  \textit{Var}(0) = 1,
\end{equation}
which has the solution $\textit{Var}(i) = 2^i$, indicating that each level-$l$ CFLOBDD represents a (Boolean-valued or non-Boolean-valued) function that takes $2 ^ i$ Boolean variables as arguments.
}

Finally, to make sure that each Boolean function \revision{that takes $2 ^ k$ Boolean variables} has a unique, canonical representation as a \revision{level-$k$} CFLOBDD, some structural invariants are enforced.
The complete definition of CFLOBDDs (structural invariants included) can be found in \Cref{Appendix: StructuralInvariants}.
\revision{
It will help to keep in mind that the goal of the invariants is to force there to be a \emph{unique} way to fold a given decision tree into a CFLOBDD that represents the same Boolean function. The decision-tree folding method is discussed in \cite[\S4.2]{TOPLAS:SCR24} and \cite[Appendix C]{TOPLAS:SCR24}, but the main characteristic of the folding method is that it works greedily, left to right.
}
When the structural invariants are respected, CFLOBDDs enjoy the following canonicity property \cite[Theorem 4.3]{TOPLAS:SCR24}:

\begin{theorem}[Canonicity]\label{The:Canonicity}
If $C_1$ and $C_2$ are level-$k$ CFLOBDDs for the same Boolean function over $2^k$ Boolean variables, and $C_1$ and $C_2$ use the same variable ordering, then $C_1$ and $C_2$ are isomorphic.
\end{theorem}

\subsection{Terminology Conventions}
\label{Se:TerminologyConventions}

In this section, we introduce the following terminological conventions:
\begin{itemize}
  \item
    Let $f$ be a function that takes $2^k$ Boolean parameters $x_0, \cdots, x_{2 ^ k - 1}$.
    In the remainder of the paper, we refer to the BDD for $f$ as $B(f)$ and the CFLOBDD for $f$ as $C(f)$. If $f$ is not specified, we refer to $B(f)$ and $C(f)$ as $B$ and $C$, respectively.
    In addition, we will assume that the same variable ordering is used in $B$ and $C$.
  \item
    We will omit the value-tuple (value-edges and the terminal values) from $|C|$.
    The size of that part of $C$ is the same as the corresponding terminal elements of $B$.
    \hspace{1.5ex}
    In addition, we do not count the vertices and edges in level-0 groupings in \Cref{Subse: Count Vertices} and \Cref{Subse: Count Edges}, respectively, because their numbers are each bounded by a constant.
    Because of the unique-representation property provided by hash-consing \cite{Tokyo-TR-74-03:Goto74},\footnote{
\revision{
      In the literature on decision diagrams, hash consing is typically not mentioned explicitly;
      however, maintaining the ''unique ids'' by means of a ''unique table'' \cite{dac:BRB90} so that equality can be checked in unit time by a single pointer comparison is exactly the technique of hash consing.
}
    }
    there is only a single copy of a fork grouping and a single copy of a don't-care grouping.
    Consequently, a CFLOBDD has at most two level-0 groupings, which collectively contribute at most five vertices and three edges.
  \item
    $|B|$ refers to the size of $B$, which can be understood either as the number of nodes in $B$ or as the number of edges (transitions): the two quantities are related by a constant factor.
  \item
    Let $\revision{\mathcal{D}_B(d)} \eqdef  |\{u \in B \mid \textit{depth}[u] = d \}|$, the number of nodes at depth $d$ in $B$. In particular, when $d=2^k$, \revision{$\mathcal{D}_B(d)$} denotes the number of leaf nodes in $B$.
    \revision{We will omit the subscript $B$ and simply use the symbol $\mathcal{D}$ when there is little danger of confusion.}
    For the BDD in \Cref{fig:BDDexample}, $\mathcal{D}(1) = 2$ and $\mathcal{D}(3) = 3$.

    \hspace{1.5ex}
    It is clear that $|B|=\sum_{i=0}^{2^k} \mathcal{D}(i)$.

  \item
    Following the convention used in \cite{TOPLAS:SCR24}, we use the term ``nodes'' solely when referring to BDDs; for CFLOBDDs, we refer to the components inside groupings as ``vertices.''
  \item
    For convenience in the counting arguments made later in the paper, we always consider call and return edges to be part of the \emph{caller} groupings.
    For example, all the call and return edges in \Cref{fig:CFLOBDDexample} are considered to be part of $g_2$.
\end{itemize}

%% file: 3_overview.tex
\section{Overview}
\label{Se:Overview}

\revision{
This section first discusses the issue ``3/4-depth duplication'' (\Cref{Subse: 3/4-depth-duplication}), which suggests that there is an obstacle to the existence of a polynomial bound on the size of CFLOBDD $C$ relative to the size of BDD $B$, and then presents intuition about a concept we call ``level locality'' (\Cref{Subse: intuitionLevelLocality}), which is fundamental to showing that $|C|$ is at most cubic in $|B|$. Level locality will be formalized in \Cref{Subse: Phi} (via \Cref{Cor:GroupingsAtLSmallerThanB}).
}

\subsection{3/4-Depth Duplication} \label{Subse: 3/4-depth-duplication}

\begin{figure}
    \centering
    \includegraphics[width=1.0\linewidth]{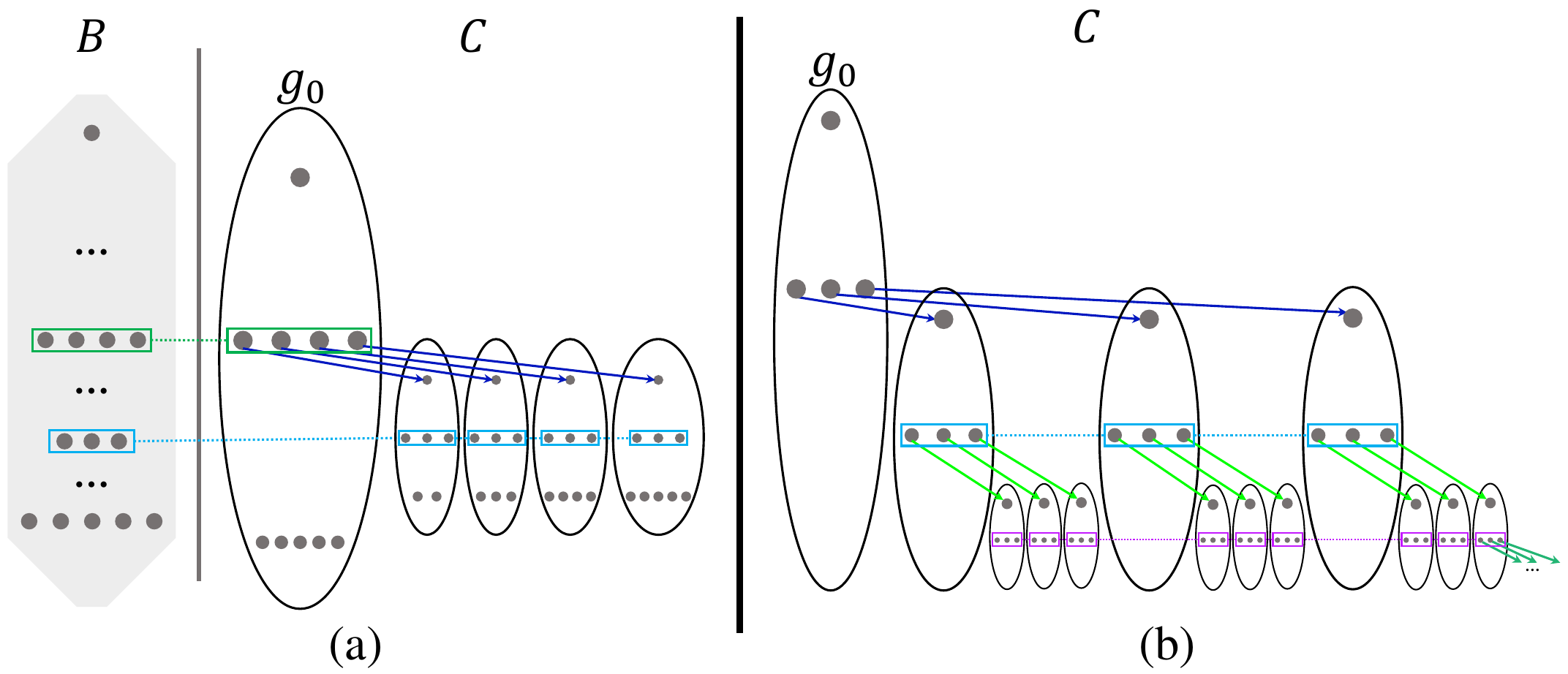}
    \caption{
\revision{
    (Color online.)
}
    (a) An illustration of ``3/4-depth duplication.''
    (b) A hypothetical scenario where ``3/4-depth duplication'' propagates, causing the size of $C$ to grow exponentially compared with the size of $B$.
\revision{
    The middle vertices enclosed in purple rectangles in (b) indicate the duplicates caused by the hypothesized propagation of ``3/4-depth duplication'' to groupings at the next-lower level, which could potentially propagate to still further lower levels.
}
    (To reduce clutter, return edges are elided in these diagrams.)
    }
    \label{fig:3/4 height problem}
\end{figure}

\revision{
As shown in \Cref{fig:3/4 height problem}(a), there are some structural}
differences between a multi-entry decision diagram (i.e., fragments in the interior of $B$) and a single-entry (hierarchical) decision diagram (i.e., proto-CFLOBDD $C$).
Let $g_0$ be the outermost grouping of $C$.
If we look at the interpretation of the second half of the variables in $B$ (in the bottom half of $B$), we find that it is a multi-entry finite-state machine (where the ``entry'' nodes are the nodes at half-depth in $B$).
In contrast, in CFLOBDDs, we will interpret the second half of the variables with some level-$(k\text{-}1)$
\revision{
proto-CFLOBDDs---possibly just one---each
}
of which is a single-entry hierarchical structure.
\revision{}{
When
}
the single-entry proto-CFLOBDDs at level-$(k\text{-}1)$ cannot share with each other, the B-callees of $g_0$ have to be different groupings,
\revision{
which leads to some duplication, as depicted in \Cref{fig:3/4 height problem}(a).
}

\revision{
By ``duplication,'' we do not mean that some grouping appears more than once in the CFLOBDD: that kind of structural repetition would violate the CFLOBDD structural invariants.
Rather, we use ``duplication'' in the narrow sense of a repetition of elements when a (proto-)CFLOBDD is compared with the corresponding BDD.
In particular, as shown in \Cref{fig:3/4 height problem}(a), the sets of middle vertices of the four ``callees'' of $g_0$, enclosed in cyan-colored boxes, each correspond to the nodes of $B$ that are enclosed in the cyan-colored box.
}

\revision{
Because $C$ has multiple copies of the cyan-boxed nodes from $B$, such duplication
raises concerns about a CFLOBDD-to-BDD size bound.
}
\revision{
In general, the nodes of $B$
}
at 3/4-depth correspond to the middle vertices of the B-callees of $g_0$, but these nodes can be duplicated many times to create the middle vertices of the different level-$(k\text{-}1)$ groupings. 
\revision{
To refer to this phenomenon concisely,
}
we call this form of duplication ``3/4-depth duplication.'' 
\revision{
We see from \Cref{fig:3/4 height problem}(a) that, in the worst case, some structures at level-($k$-1)---typically vertices---can be duplicated $\mathcal{D}(2 ^ {k - 1})$ times (i.e., the number of BDD nodes at half depth).
}

\revision{Moreover, because the relationship between multi-entry decision diagrams and single-entry (hierarchical) decision diagrams discussed above could hold at every level, it is possible for ``3/4-depth duplication'' to occur at every level of a proto-CFLOBDD.
}
What is worrisome about this phenomenon is that if such duplication 
\revision{
were to \emph{cascade} across many (or most) levels of a CFLOBDD, $|C|$ could end up being exponentially larger than $|B|$.
\Cref{fig:3/4 height problem}(b) shows a hypothetical scenario for such a cascade, where the groupings branch in a tree-like manner when going from higher levels to lower levels.
} 

\revision{
If such a cascade were to occur, for each $l$, the sizes of the structures at level $l-1$ could be $\mathcal{D}(d_{i_l})$ times larger than the size of structures at level $l$ (for some depth $d_{i_l}$), in which case the size of $C$ could include a product of $\Theta(k)$ of such $\mathcal{D}(d_{i_l})$ factors, which grows faster than any polynomial in $|B|$.
Consequently,
}
it is natural to be concerned about the following question:

\begin{quote}
    Does ``3/4-depth duplication'' at one
\revision{
    level
}
    propagate to \revision{lower} levels?
\end{quote}

\noindent
\revision{Fortunately, it turns out that although ``3/4-depth duplication'' can occur at every level, there is never a cascade effect, as will be established by the counting argument in \Cref{Se:Relations} and \Cref{Se:Counting}.}

\subsection{Intuition about ``Level-Locality''} \label{Subse: intuitionLevelLocality}

\revision{We now present an different viewpoint under which the number of groupings in $C$ can be seen to be at most a polynomial in $|B|$, providing an intuitive basis for the counting argument formalized in \Cref{Se:Relations} and \Cref{Se:Counting}.}

\revision{
Let $g$ be a grouping at level $l$.
There is at least one partial assignment $\alpha = [x_0 \mapsto a_0, \cdots, x_{j-1} \mapsto a_{j-1}]$ that corresponds to a matched path in $C$ ending at the entry vertex of $g$.
If we interpret $\alpha$ in $B$, we get to a node $n$ at depth $j$. Let $f_n$ be the residual function represented by $n$, which takes $2^k - j$ Boolean variables $x_j, \cdots, x_{2^k - 1}$. 
Any residual matched path in $C$, which starts from the entry node of $g$ and reaches one of the value vertices of $C$ after interpreting $x_j, \cdots, x_{2^k - 1}$, corresponds to an assignment of $x_j, \cdots, x_{2^k - 1}$ for the Boolean function $f_n$.
Among these variables, $x_j, \cdots, x_{j + 2^l - 1} $ are interpreted in the proto-CFLOBDD headed by $g$, so in some sense, the proto-CFLOBDD headed by $g$ interprets the first $2^l$ variables of $f_n$.
The number of residual functions like $f_n$ is $|B|$ (because each node of $B$ corresponds to a residual function, and there are $|B|$ nodes in total).
Therefore, there can be at most $|B|$ groupings at level $l$.
The total number of groupings can then be no more than $k \times |B|$.
}

\revision{
This viewpoint analyzes the number of level-$l$ groupings directly from the structure of $B$ without considering the other levels, thereby achieving some form of ``level locality'': the number of groupings at any given level $l$ is bounded by $|B|$.
}

\revision{
Actually, ``level locality'' is fundamental to showing that $|C|$ is cubic in $|B|$.
We will establish the bound following the intuition above: 
\Cref{Se:Relations} defines a binary relation ``$\triangleright$'' between the nodes in $B$ and the groupings in $C$ based on semantic notions such as partial assignments and paths, and then proves that this binary relation is many-one, thereby establishing a structural relationship between $B$ and $C$. 
\Cref{Se:Counting} presents upper bounds on the number of groupings, vertices, and edges of $C$ as a function of $|B|$ by counting them at each level independently and summing the results.
``3/4-depth duplication'' can happen at every level, but it does 
\revision{
\emph{not} propagate, and hence
}
\Cref{fig:3/4 height problem}(b) is a misleading picture.
\Cref{Se:TightInstances} presents
\revision{
tight instances of all three bounds, the construction of which is based on ``3/4-depth duplication.''
}
In \Cref{Se:Conclusion}, we recap why ``3/4-depth duplication'' does not propagate across levels.
}

%% file: 4_relations.tex
\section{Structural Relationships between B and C}
\label{Se:Relations}

This section establishes two structural relationships between $B$ and $C$. \Cref{Subse: relation |>} defines a binary relation ``$\triangleright$'' between nodes of $B$ and groupings of $C$. \Cref{Subse: Properties |>} proves some properties of ``$\triangleright$''. Based on that, \Cref{Subse: Phi} defines a mapping $\Phi$, which captures a key relationship between $B$ and $C$.

\subsection{The Binary Relation ``$\triangleright$''}
\label{Subse: relation |>}

\revision{
As mentioned in \Cref{Subse: intuitionLevelLocality}, to establish how the structures of $B$ and $C$ are related, 
}
we start with a relation ``$\triangleright$'' between nodes of $B$ and 
\revision{
groupings of $C$.
Informally, $n \triangleright g$ captures the following condition:
\begin{quote}
  There exists a partial assignment $\alpha$ whose interpretation in $B$ follows a path to node $n$, and whose interpretation in $C$ follows a matched path to the entry vertex of $g$.
\end{quote}

(This notion will be defined formally in \Cref{Def:Simu}.)
}

\begin{example}\label{Ex:TriangleManyMany}

\begin{figure}
    \centering
    \includegraphics[width=0.4\linewidth]{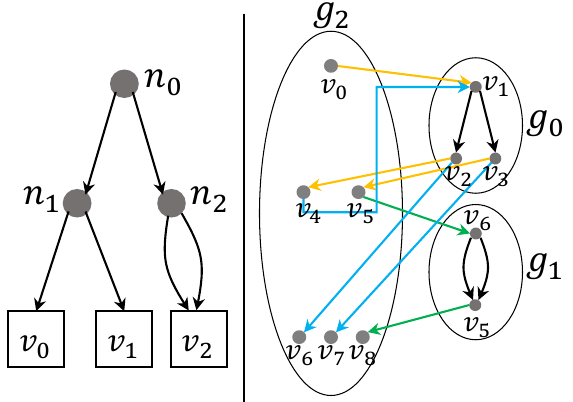}
    \caption{
\revision{
    (Color online.)
}
    The BDD (left) and the CFLOBDD (right) for the function ``$\textit{if}\ (\neg x_0)\ \textit{then}\ (\textit{if}\ (\neg x_1)\ \textit{then}\ v_0\ \textit{else}\ v_1)\ \textit{else}\ v_2$.''}
    \label{fig:DefSimu}
\end{figure}

The BDD and the CFLOBDD for the function
\[
  \textit{if}\ (\neg x_0)\ \textit{then}\ (\textit{if}\ (\neg x_1)\ \textit{then}\ v_0\ \textit{else}\ v_1)\ \textit{else}\ v_2
\]
is shown in \Cref{fig:DefSimu}.
When we interpret the partial assignment $[x_0 \mapsto 0]$, we will take the path $n_0 \xrightarrow{x_0 \mapsto 0} n_1$ in B and $v_0 \xrightarrow{\text{call}} v_1 \xrightarrow{x_0 \mapsto 0} v_2 \xrightarrow{\text{ret}} v_4 \xrightarrow{\text{call}} v_1$ in $C$. 
Consequently, we have $n_1 \triangleright g_0$.
In addition, via the empty partial assignment [], $n_0 \triangleright g_2$ and $n_0 \triangleright g_0$ both hold.

\end{example}

The following is the formal definition of the binary relation ``$\triangleright$'':

\begin{definition}\label{Def:Simu}
    For a node $n \in B$ and a grouping $g \in C$, $n \triangleright g$ holds iff there is an assignment $\alpha$ and some $j$, such that:
    \begin{itemize}
      \item
        we reach $n$ after interpreting $\alpha(x_0), \alpha(x_1), \cdots, \alpha(x_{j-1})$ in $B$, and
      \item
        we reach the entry vertex of $g$ after the following steps:
        \begin{enumerate}
          \item 
            \label{It:InterpretStep}
            interpreting $\alpha(x_0), \alpha(x_1), \cdots, \alpha(x_{j-1})$ in $C$,
\revision{
            following a path $P$ that obeys the Matched-Path Principle of \Cref{Se:CFLOBDDs}
}
          \item
            \label{It:FollowMatchingReturnsStep}
            following some number of return edges (possibly zero),
\revision{
            obeying the Matched-Path Principle with respect to closest preceding call edges in $P$
}
          \item
            \label{It:FollowCalsStep}
            following some number of call edges (possibly zero), reaching the entry vertex of $g$.
          \end{enumerate}
    \end{itemize}
\end{definition}

\revision{In \Cref{Ex:TriangleManyMany}, the partial assignment $[x_0 \mapsto 0]$ witnesses $n_1 \triangleright g_0$. According to \Cref{Def:Simu}, the corresponding path is divided into three parts, labeled (1), (2), and (3), as follows:

\[
  \underbrace{v_0 \xrightarrow{\text{call}} v_1 \xrightarrow{x_0 \mapsto 0} v_2}_{(1)} 
  \underbrace{\xrightarrow{\text{ret}} v_4}_{(2)}
  \underbrace{\xrightarrow{\text{call}} v_1}_{(3)}
\]

The empty partial assignment $[]$ witnesses $n_0 \triangleright g_0$. In this case, the matched path in $C$ consists of just a single call edge from $v_0$ to $v_1$.
According to \Cref{Def:Simu}, this edge belongs to the (3) part, and the (1) and (2) parts are empty:
\[
  \underbrace{v_0 \xrightarrow{\text{call}} v_1 }_{(3)}
\]

The empty assignment also witnesses $n_0 \triangleright g_2$.
In this case, the matched path in $C$ is the degenerate path consisting of just the vertex $v_0$, so parts (1), (2), and (3) are all empty.

We note that the (1) and (2) parts of the matched path in $C$ are empty only when $n$ is the root node of $B$; the (3) part of the path is empty only when $n$ is the root node of $B$ and $g$ is the outermost grouping of $C$.
In addition, in \Cref{Ex:TriangleManyMany}, $n_0 \triangleright g_0$, $n_0 \triangleright g_2$, and $n_1 \triangleright g_0$ all hold. In this sense, the relation ``$\triangleright$'' is many-many.
}


\subsection{The Properties of ``$\triangleright$''}
\label{Subse: Properties |>}
 
We hope to associate every grouping in $C$ with some components of $B$. As we now show, ``$\triangleright$'' covers every grouping in $C$, and it is strongly connected with depth-alignment in $B$. The following theorem formalizes these properties:

\begin{theorem} \label{Lem:Align}
\revision{
    For a given $l$, the relation $\triangleright$ is ``surjective'' in both directions between the level-$l$ groupings and the nodes at depths\footnote{
\revision{
      Recall that the depth of a node equals the number of Boolean variables that have been interpreted when the node is reached, and for a BDD with $2 ^ k$ variables, a node at depth $d$ corresponds to a residual function with $2 ^ k - d$ variables. When discussing divisibility by $2 ^ i$ (with $0 \le i \le k$) in this theorem, using $d$ or $2 ^ k - d$ gets equivalent results, because $d$ and $2 ^ k - d$ share the same divisibility for $2^i$ (with $0 \le i \le k$).
}
    }
    that are multiples of $2 ^ l$:
}
    \begin{enumerate}[label=(\alph*)]
        \item For any level-$l$ grouping $g$, there exists $n_0 \in B$ such that $n_0 \triangleright g$ and \depth[$n_0$] is a multiple of $2^l$. 
        \item For any node $n$, if \depth[n] is a multiple of $2^l$ and $n$ is not a terminal node, there exists a level-$l$ grouping $g_0$ such that $n \triangleright g_0$.
    \end{enumerate}
\end{theorem}

\begin{proof}
    We first prove $(a)$. We can assert that there is a partial assignment $\alpha$ that leads us to the entry vertex of $g$ in $C$ (following steps (1), (2), and (3) from \Cref{Def:Simu}). There can be many such $\alpha$s, but it suffices to pick any of them. We will get to some node $n_0$ if we interpret $\alpha$ in $B$. $n_0 \triangleright g$ obviously holds. Moreover, the depth of node $n_0$ must be a multiple of $2^l$. To prove it, we fix $l$ and do an induction on $k$. Recall that $k$ is the maximum level of $C$, so it is obvious that $k \ge l$. For this induction, the base case is $k = l$ because the premises are contradictory when $k < l$. Because $\depth$ in $B$ is defined as the number of variables read, we use the following inductive hypothesis:
    \begin{quote}
        For a proto-CFLOBDD with maximum level $k$, if we reach the entry vertex of a level-$l$ grouping with the partial assignment $[x_0 \mapsto a_0, \cdots, \revision{x_{j-1} \mapsto a_{j-1}}]$, then the number of variables read \revision{(i.e., $j$)} is a multiple of $2^l$.
    \end{quote}
    The cases are shown in \Cref{Fig: Align}. Let $g_0$ be the outermost grouping of $C$:

    \begin{figure}
        \centering
        \includegraphics[width=0.5\linewidth]{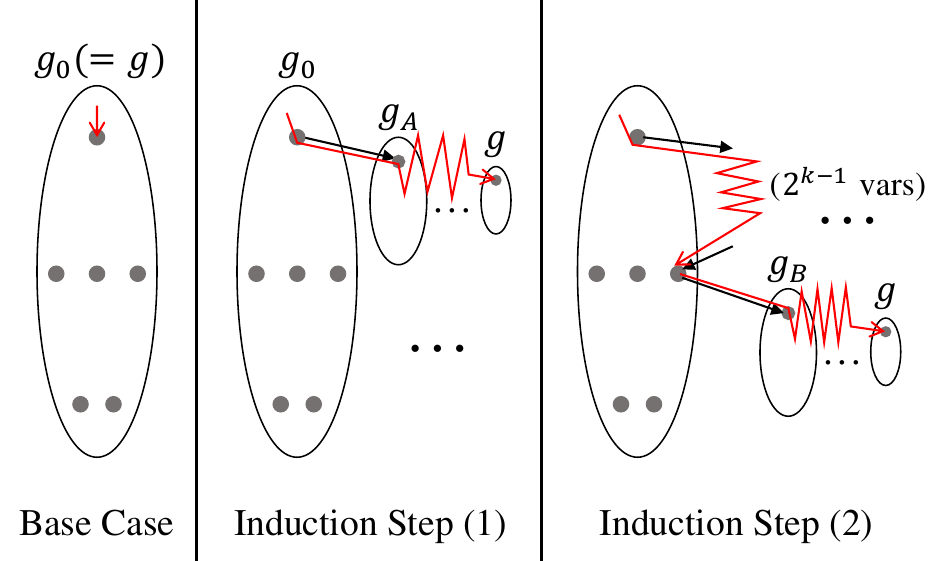}
        \caption{
\revision{
        (Color online.)
}
        The cases for the proof of \Cref{Lem:Align}. The red line stands for 
\revision{
        the path considered
}
        in CFLOBDD $C$.
\revision{
        (To reduce clutter, return edges are elided in these diagrams.)
}
        }
        \label{Fig: Align}
    \end{figure}

    \paragraph{Base Case:} If $k = l$ , then $g = g_0$; $n$ must be the root of $B$. No variable has been read, and the result
    holds because 0 is a trivial multiple of $2^k$.

    \paragraph{Induction Step:}
    Assume that the inductive hypothesis is true for (proto-CFLOBDDs with maximum level) $k-1$, and now $k > l$. Our goal is to show that the induction hypothesis is true for $k$. There are two cases for $g$'s ``position'' in $C$ after following the path for the partial assignment $[x_0 \mapsto a_0, \cdots, x_j \mapsto a_j]$:
    \begin{enumerate}
      \item
        If we are still in the A-connection part of $g_0$
        (that is, the path did not reach the middle vertices of $g_0$), all readings of the values of variables
        take place in the proto-CFLOBDD headed by $g$'s A-callee $g_A$, which is a proto-CFLOBDD of maximum level $k - 1$. We can draw the conclusion directly by applying the inductive hypothesis to $g_A$, and thus the induction hypothesis is true for $k$.
      \item
        If we are in the B-connection part, we divide the path into two parts: the part in the A-connection (the proto-CFLOBDD headed by $g_A$) of $g$, and the part in some B-connection (a proto-CFLOBDD headed by $g_B$) of $g$. We read $2 ^ {k-1}$ variables in the first part. According to the inductive hypothesis, the number of variables read in the second part is a multiple of $2^l$. Because $2 ^ {k-1}$ is a multiple of $2 ^ l$, their sum is a multiple of $2 ^ l$, and thus the induction hypothesis is true for $k$.
    \end{enumerate}
    We next prove $(b)$. Similar to $(a)$, we can use the following inductive hypothesis:
    \begin{quote}
        For a proto-CFLOBDD with maximum level $k$, if $j$ is a multiple of $2^l$ and $j < 2 ^ k$, 
        then for any partial assignment
\revision{
        $\alpha = [x_0 \mapsto a_0, \cdots, x_{j-1} \mapsto a_{j-1}]$,
}
        there exists some $g$ at level-$l$ such that we reach the entry vertex of $g$ after: 
        \begin{enumerate}
          \item 
            interpreting $\alpha(x_0), \alpha(x_1), \cdots, \alpha(\revision{x_{j-1}})$ in $C$
          \item
            following some number of return edges (possibly zero)
          \item
            following some number of call edges (possibly zero)
          \end{enumerate}
    \end{quote}

    \Cref{Fig: Align} also shows the cases for the proof of $(b)$. We do an induction with the same pattern as \textit{(a)}. Again depth in $B$ corresponds to the number of variables read. Let $g_0$ be the outermost grouping of $C$:
    \paragraph{Base Case:} If $k = l$, then $j$ must equal to $0$ because $0$ is the only non-negative multiple of $2^l$ less than $2^k\ (=2^l)$. Then $g_0$ is the grouping we are looking for---we reach the entry vertex of $g_0$ after interpreting $0$ variables, and following no edges.

    \paragraph{Induction Step:}
    Assume that the inductive hypothesis is true for (proto-CFLOBDDs with maximum level) $k - 1$, and now $k > l$. Our goal is to show that the induction hypothesis is true for $k$. There are two cases for the value of $j$: $0 \le j < 2 ^ {k - 1}$ and $2 ^ {k-1} \le j < 2 ^ k$.
    \begin{enumerate}
      \item
        If $0 \le j < 2 ^ {k - 1}$, the interpretation of \revision{$x_0, \cdots, x_{j-1}$} takes place in the proto-CFLOBDD headed by $g_0$'s A-callee $g_A$ after we move from $g_0$'s entry vertex to $g_A$'s entry vertex along a call edge. We can find the proper $g$  by applying the inductive hypothesis to $g_A$.
      \item
        If $2 ^ {k-1} \le j < 2 ^ k$, we will ``return'' to a middle vertex of $g_0$ after interpreting \revision{$x_0, \cdots, x_{2^{k-1} - 1}$} in the proto-CFLOBDD headed by $g$'s A-callee. Then, we will interpret the remaining $j - 2 ^ {k-1}$ variables (\revision{$x_{2^{k-1}}, \cdots, x_{j-1}$}, possibly zero) in one of $g_0$'s B-callees $g_B$. Because both $j$ and $2^{k-1}$ are multiples of $2^l$, $j - 2 ^ {k-1}$ is a multiple of $2 ^ l$. Then we can apply the inductive hypothesis to $g_B$ and find the proper $g$, and thus the induction hypothesis is true for $k$.
    \end{enumerate}
\end{proof}

\Cref{Lem:Align}(a) shows the possibility to make a mapping from some components of $B$ to $C$. To make a well-defined mapping, we also have to show that ``$\triangleright$'' is many-one when looked at in the right way. As we saw in \Cref{Ex:TriangleManyMany}, the ``$\triangleright$'' relation itself is many-many. But actually, ``$\triangleright$'' is many-one with the level-$l$ fixed. To establish that many-oneness holds after fixing the level, as well as to serve counting purposes later, we need to establish correspondences among  ``bigger'' structures in BDDs and CFLOBDDs. In a CFLOBDD, a ``bigger structure'' is exactly a proto-CFLOBDD \cite[Def.\ 4.1]{TOPLAS:SCR24}, whereas in the corresponding BDD, a ``bigger structure'' is what we call a \emph{BDDpatch}:

\begin{definition}\label{De:BDDpatch}
    For a
\revision{
    BDD
}
    node $n$, BDDpatch$(n, h)$ is defined as the substructure containing node $n$ and all descendants of \revision{$n$} at depths \depth{[\revision{$n$}]} to \depth[\revision{$n$}]$+ h \le 2^k$ (inclusive) \revision{and all the edges between them}. 
\end{definition}

For the BDD in \Cref{fig:DefSimu}, BDDpatch$(n_0, 1)$ contains nodes $\{n_0, n_1, n_2\}$ and \revision{{the edges $\{n_0 \to n_1, n_0 \to n_2\}$}}.
BDDpatch$(n_2, 1)$ contains nodes $\{n_2, v_2\}$ and the two edges from $n_2$ to leaf $v_2$. Note that BDDpatch$(n, h)$ can also be viewed as a BDD---the ``value'' of each leaf is a pseudo-Boolean function over the variables $\{ x_{\depth[\revision{n}]+h}, \ldots, x_{2^k-1} \}$.
From the vantage point of $B$, let $f_n$ be the residual function represented by $n$, which takes variables $\{ x_{\depth[n]}, \ldots, x_{2^k-1}\}$.
Suppose that BDDpatch($n, h$) has $m$ leaves.
Function $f_n$ can be considered to be the composition of
(i) the function over $h$ variables represented by BDDpatch($n, h$) proper, with (ii) $m$ (different) functions $h_1$, $\ldots$, $h_m$---each over the variables $\{ x_{\depth[n]+h}, \ldots, x_{2^k-1}\}$---corresponding to the leaf nodes of BDDpatch($n, h$).

We can now establish a correspondence between BDDpatches and proto-CFLOBDDs. Let $\alpha = [x_0 \mapsto a_0, x_1 \mapsto a_1, \cdots, x_j \mapsto a_j]$ be a partial assignment to the first $j+1$ variables, and suppose that by interpreting $\alpha$ we get to node $n$ in $B$, and the entry vertex of grouping $g$ in $C$. Then variables $x_{j+1}, x_{j+2}, \cdots, x_{j+2^l}$ will be interpreted both (a) during a traversal of BDDpatch($n, 2^l$), and also (b) during a traversal of the proto-CFLOBDD headed by $g$. In this sense, the proto-CFLOBDD headed by $g$ should ``simulate'' BDDpatch($n, 2^l$). \Cref{Lem:Corr} formalizes this intuition:

\begin{lemma}\label{Lem:Corr}
    Let $g$ be a level-\textit{l} grouping in $C$. For a node $n$ in $B$, if $n \triangleright g$, then BDDpatch$(\revision{n}, 2^l)$ is equivalent to the proto-CFLOBDD headed by $g$. That is, there exists a bijection\footnote{
      There is not an explicit order to the leaves of a BDD. We index the leaves of the BDDpatch according to the lexicographic order of the leftmost reaching path, and the bijection is an identity map (between the BDDpatch's leaf node $i$ and the grouping's exit vertex $i$).
    }
    between the leaf nodes of BDDpatch$(\revision{n}, 2^l)$ and the exit vertices of $g$ such that we always get to a corresponding leaf node/exit vertex pair for any assignment of $2^l$ variables.
\end{lemma}

The proof can be found in the \Cref{Appendix: ProofLemCorr}.

With the help of \Cref{Lem:Corr}, we can establish a condition under which the many-oneness holds for ``$\triangleright$'':
\begin{lemma}\label{Lem:ManyOne}
    For all $n \in B$, if $g$ and $g'$ are two level-$l$ groupings for which $n \triangleright g$ and $n \triangleright g'$ both hold, then $g \equiv g'$ (i.e., $g$ is identical to $g'$).
\end{lemma}

\begin{proof}
    According to \Cref{Lem:Corr}, $g$ and $g'$ are both equivalent to BDDpatch($\revision{n}, 2^l$), which means that their corresponding exit vertices are reached (from the entry vertices of $g$ and $g'$) by the same sets of assignments.\footnote{
      An alternative way of thinking about $g$ and $g'$ is to consider an assignment as a binary string in $\{ 0,1 \}^{2^l}$, and thus each exit vertex of a level-$l$ proto-CFLOBDD represents a non-empty set of strings---i.e., a language, and the languages for the different exit vertices partition the set $\{ 0,1 \}^{2^l}$.
      See \cite[\S6]{cflobdd_arxiv}.
    }
    By the CFLOBDD Canonicity Theorem \cite[Thm.\ 4.3]{TOPLAS:SCR24}, $g$ is identical to $g'$.
\end{proof}

\revision{
\Cref{Lem:Corr} also helps us establish a result that relates the middle- and exit-vertices in $C$ with some nodes in $B$:
}

\begin{figure}
    \centering
    \includegraphics[width=0.6\linewidth]{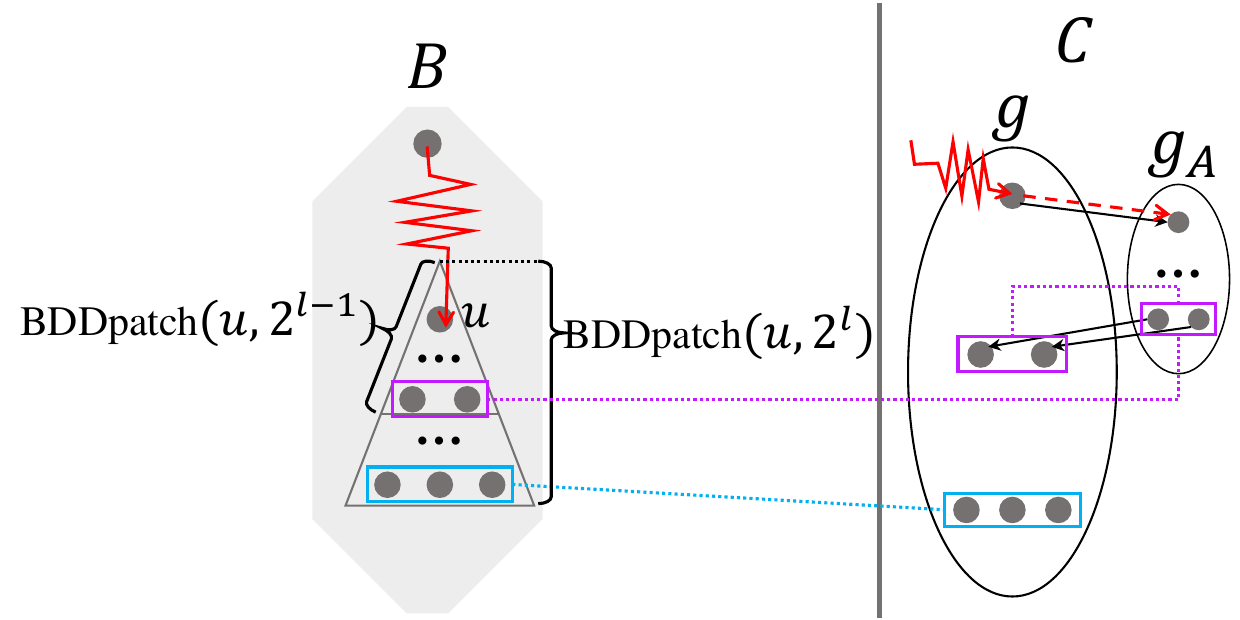}
    \caption{
    \revision{(Color online.)}
    An illustration of the properties established in \Cref{Lem:NumberMidExit1New}.
}
    \label{Fig:NumberMidExitNew}
\end{figure}

\revision{
\begin{lemma}\label{Lem:NumberMidExit1New}
    Let $u$ be a node in $B$, and let $g$ be a level-$l$ grouping in $C$. If $u \triangleright g$, then:
    \begin{enumerate}[label=(\alph*)]
        \item if $l \ge 1$, the number of middle vertices of $g$ equals the number of leaf nodes of BDDpatch$(u, 2 ^ {l - 1})$.
        \item the number of exit vertices of $g$ equals the number of leaf nodes of BDDpatch$(u, 2 ^ l)$.
    \end{enumerate}
\end{lemma}
}

\revision{
 An illustration of the properties established in \Cref{Lem:NumberMidExit1New} is given in \Cref{Fig:NumberMidExitNew}.
}

\revision{
\begin{proof}
    We first prove \textit{(b)}. According to \Cref{Lem:Corr}, $g$ is equivalent to BDDpatch($u, 2^l$) (i.e., there is a bijection between the leaf nodes of BDDpatch($n, 2 ^ l$) and the exit vertices of $g$ about paths. Such a bijection implies the equality between the number of exit vertices of $g$ and the leaf nodes of BDDpatch$(u, 2 ^ l)$.

    We can establish property \textit{(a)} by changing our viewpoint to level $(l\text{-}1)$.
    Because $l \ge 1$, $g$ is a non-level-0 grouping, we can move from $g$'s entry vertex to the entry vertex of $g$'s A-callee $g_A$ (via $g$'s A-call edge) without interpreting any Boolean variable.
    According to the definition of $\triangleright$, this step is just appending a call edge into the (3) part of our path (see \Cref{Def:Simu}), and thus we have $u \triangleright g_A$. 
    Because $g_A$ is at level $(l\text{-}1)$, we know from \textit{(b)} that the number of exit vertices of $g_A$ equals the number of leaf nodes of BDDpatch$(u, 2 ^ {l -1})$.
    The CFLOBDD structural invariants (\Cref{De:CFLOBDD}) tell us that the return-map of $g$'s A-connection is the identity mapping. 
    Therefore, the number of exit vertices of $g_A$ equals the number of middle vertices of $g$.
    By transitivity, the number of middle vertices of $g$ equals the number of leaf nodes of BDDpatch$(u, 2 ^ {l - 1})$.
    The three sets of equal vertices are indicated by the purple boxes in \Cref{Fig:NumberMidExitNew}).
\end{proof}
}

\subsection{The Mapping $\Phi$}
\label{Subse: Phi}

We now can define a mapping $\Phi$ that maps some components of $B$ to the groupings of $C$.
As shown in the previous sections, the relation ``$\triangleright$'' itself is not many-one, but it becomes many-one if we fix the level $l$. Thus, we define the following set as the domain of $\Phi$:

\begin{definition}\label{Def:NS}
    (Node Stratification) Let $\NS(B) = \{ (u, l) \mid \revision{u \in B \land u\text{ is not a leaf node} \land} \depth[u]$ is a multiple of $2^l\}$ 
\end{definition}

\begin{figure}
    \centering
    \includegraphics[width=1.0\linewidth]{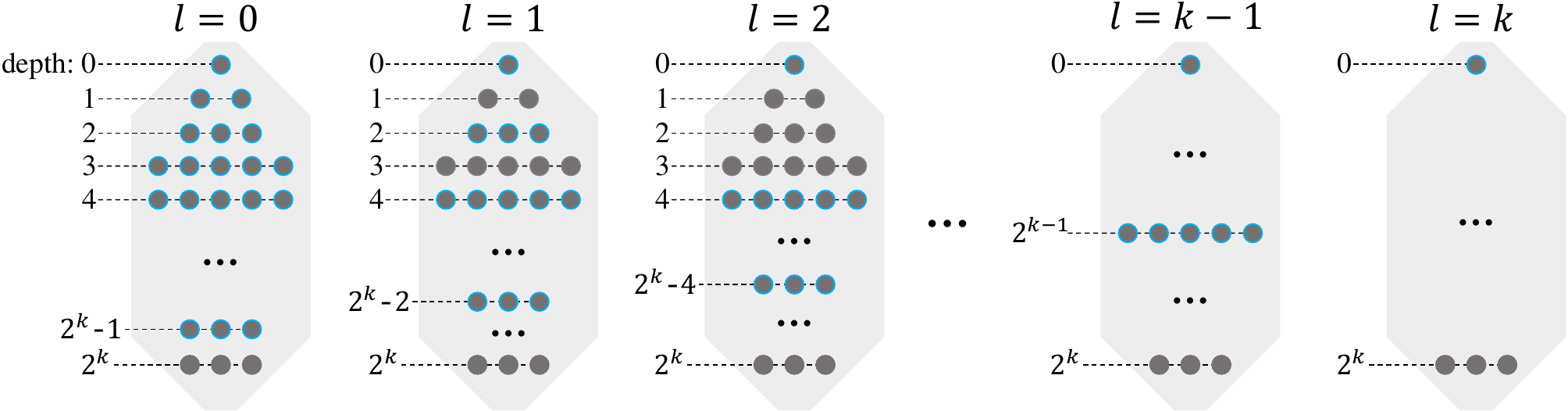}
    \caption{
      \revision{ (Color online.)
      An illustration of $\NS(B)$ from \Cref{Def:NS}. The set of blue-marked nodes in the respective diagrams for the different values of $l$ indicate the subset of $\{ (u, l) \in \NS(B) \}$ for that value of $l$.
      }
    }
    \label{Fig:IteratePatternOvercount}
\end{figure}

\revision{
The different ``strata'' of $\NS(B)$ are depicted in \Cref{Fig:IteratePatternOvercount}.
}
Note that a given node $n$ can be part of many elements of $\NS(B)$. For instance, if \depth[$n$] = 24, then $(n, 3)$, $(n, 2)$, $(n, 1)$ and $(n, 0)$ are all elements of \NS(B).

We define the mapping $\Phi$ from $\NS(B)$ to the groupings of $C$ as follows:

\begin{definition}\label{def:phi}
  Let $G(C)$ denote the set of groupings of $C$. We define a mapping $\Phi$ from $\NS(B)$ to $G(C)$ according to ``$\triangleright$'':
  \begin{quote}
    $\Phi$ maps $(u, l)$ to the level-$l$ grouping $g$ such that $u \triangleright g$.
  \end{quote}
\end{definition}

We first have to show that the mapping $\Phi$ is well-defined. Given $(u, l) \in \NS(B)$, \Cref{Lem:Align}(b) tells us that there exists some level-$l$ grouping $g \in G(C)$ such that $u \triangleright g$; \Cref{Lem:ManyOne} shows that such a level-$l$ grouping $g$ is unique. \revision{ Therefore, $\Phi$ selects a unique grouping in $G(C)$ for every $(u, l) \in \NS(B)$. }

The mapping $\Phi$ captures a key structural relationship: 

\begin{theorem}\label{trm:phi}
    Let $\Phi[S] \eqdef \{ \Phi(u, l) \mid (u, l) \in S \}$, where $S \subseteq \NS(B)$. Then $G(C) = \Phi[\NS(B)]$ (i.e., $\Phi$ is an onto-mapping from $\NS(B)$ to $G(C)$).
\end{theorem}
\begin{proof}
    The well-definedness of $\Phi$ implies that $\Phi[\NS(B)] \subseteq G(C)$.
    We now need to show that $G(C) \subseteq \Phi[\NS(B)]$.
    \Cref{Lem:Align}(a) tells us that for any level-$l$ grouping $g$, there exists a non-terminal node $u \in B$ such that $u \triangleright g$ and $\depth[u]$ is a multiple of $2 ^ l$. According to the definition of $\NS$, $(u, l)$ must be in $\NS(B)$. Consequently, $g \in \Phi[\NS(B)]$, and thus $G(C) \subseteq \Phi[\NS(B)]$.
\end{proof}

$\Phi$ could be viewed as the mapping that maps some components of BDD $B$ to the part of CFLOBDD $C$ that simulates it. For the example, in \Cref{fig:DefSimu}, $\Phi(n_1, 0) = g_1$ and $\Phi(n_2, 0) = g_0$. Note that $\Phi$ is not necessarily bijective because $\Phi(n_1, 0) = g_1$ and $\Phi(n_0, 0) = g_1$ both hold. It should be intuitive due to the mechanism of ``reusing'' or ``sharing'' of the CFLOBDD.

The mapping $\Phi$ can help us figure out certain properties of $C$ by looking at $B$. In particular, it will help us do the counting in \Cref{Se:Counting} and establish the tight instances in \Cref{Se:TightInstances}. A taste of such results is provided by the following corollary, \revision{which formalizes the intuition of ``level locality'' presented in \Cref{Subse: 3/4-depth-duplication}}:

\begin{corollary}\label{Cor:GroupingsAtLSmallerThanB}
Let $G_l(C)$ denote the set of groupings of CFLOBDD $C$ at level $l$.
Then $|G_l(C)| \leq |B|$.
\end{corollary}
\begin{proof}
Let $N(B)$ denote the set of nodes of $B$, and let $\NS_l(B) \eqdef \{ u \mid (u,l) \in \NS(B) \}$.
$\Phi$ maps $\NS_l(B)$ onto $G_l(C)$  (i.e., $|\NS_l(B)| \geq |G_l(C)|$).
However, for each $l$, $\NS_l(B) \subseteq \NS_0(B) = N(B)$, and thus $|G_l(C)| \leq |\NS_l(B)| \leq |N(B)| = |B|$.
\end{proof}

%% file: 5_counting.tex
\section{The Counting Results}
\label{Se:Counting}

This section presents the upper bounds of the number of groupings (\Cref{Subse: Count Groupings}), vertices (\Cref{Subse: Count Vertices}), and edges (\Cref{Subse: Count Edges}) as a function of $|B|$, the size of $B$. We will be counting $|C|$, the size of $C$, with the mapping $\Phi$ introduced in \Cref{trm:phi}. In particular, we do an overcounting with the following formula:

\begin{formula}\label{Fml:Overcount}
We are actually interested in three different ``size'' quantities, all denoted by $|C|$, where which size measure is intended will be clear from context.

In \S4.1, \S4.2, and \S4.3, the ``size operator'' on groupings---also denoted via $|\cdot|$ in each of the right-hand sides below---denotes the number of groupings, vertices, and edges, respectively.
For each of the three versions of $|C|$, the following derivation holds:
\begin{align}
    |C| & = |\Phi[\NS(B)]|  &\text{\revision{(Rewrite using \Cref{trm:phi})}} \notag \\
    &= \left|\bigcup_{l=0}^k \bigcup_{\substack{u \in B \\ u \text{ is non-terminal} \\ \depth[u] \text{ is a multiple of } 2^l}} \Phi(u, l)\right|  \notag \\
    && \text{(overcount by assuming $\Phi$ to be bijective)} \notag \\
    &\le \sum_{l=0}^k  \sum_{\substack{u \in B \\ u \text{ is non-terminal} \\ \depth[u] \text{ is a multiple of } 2^l}} \left| \Phi(u, l) \right| \notag \\
    && \text{(refine the second ``$\Sigma$'' by enumerating the node depths)} \notag \\
    &= \sum_{l=0}^k  \sum_{d = 0}^{2^{k-l} - 1} \sum_{\substack{u \in B \\ \depth[u] = d \times 2^l}} \left| \Phi(u, l) \right|   \label{Eq:SizeCToBeReinterpreted}
\end{align}    
\revision{
The iteration pattern in \Cref{Eq:SizeCToBeReinterpreted} can also be understood in terms of the blue-marked nodes in \Cref{Fig:IteratePatternOvercount}: for each level $l$, each blue-marked node $u$ contributes to the value of $|\Phi(u, l)|$.
}

\end{formula}


\subsection{The Number of Groupings}
\label{Subse: Count Groupings}


We first give a bound on the number of groupings:

\begin{theorem} \label{Trm:Groupings}
    \revision{Let $B$ and $C$ be the BDD and CFLOBDD, respectively, for the same Boolean function $f$ (with the same variable ordering $Ord$), and} let $|G(C)|$ be the number of groupings in $C$. Then $|G(C)|=\revision{\mathcal{O}}(|B|\log |B|)$.
\end{theorem}
\begin{proof}
    This result is straightforward from \Cref{Cor:GroupingsAtLSmallerThanB}: the number of groupings at each level is no greater than $|B|$, and there are $k \revision{~= \log 2 ^ k } \le \log|B|$
\revision{
    levels.
    ($2^k \le |B|$ holds because each ply contains at least one node, and there are $2^k + 1$ plies in total.)
    Consequently,
}
    there are no more than $|B|\log|B|$ groupings in total.

    We now prove the result in a different way, via \Cref{Fml:Overcount}.
    This proof serves as a warm-up exercise that uses the mathematical techniques that will be employed shortly in bounding the number of vertices (\Cref{Trm:Vertices}) and edges (\Cref{Trm:Edges}).
    In this proof, the ``size operator'' ($|\cdot|$) in \Cref{Eq:SizeCToBeReinterpreted} counts the number of groupings, and hence is just constant function with value 1 when applied to a single grouping ($|\Phi(u, l)| = 1$). Thus, we have:
    
    \begin{align*}
        |G(C)|
        &\le \sum_{l=0}^k  \sum_{d = 0}^{2^{k-l} - 1} \sum_{\substack{u \in B \\ \depth[u] = d \times 2^l}} 1 \\
        &= \sum_{l=0}^k  \sum_{d = 0}^{2^{k-l} - 1} \mathcal{D}(d \times 2 ^ l) \\
        & \text{(Overcount by including nodes at \emph{any depth} for each level-$l$, rather than just }\\
        & \text{the depths that are multiples of $2 ^ l$ )}\\
        &\le \sum_{l=0}^k  \sum_{d = 0}^{ 2^k } \mathcal{D}(d) \\
        &= k \times |B| \\
        &\text{(because $k \le \log|B|$)}\\
        &= \revision{\mathcal{O}}(|B|\log|B|) 
    \end{align*}      
\end{proof}

\subsection{The Number of Vertices}
\label{Subse: Count Vertices}

\begin{figure}
    \centering
    \includegraphics[width=0.4\linewidth]{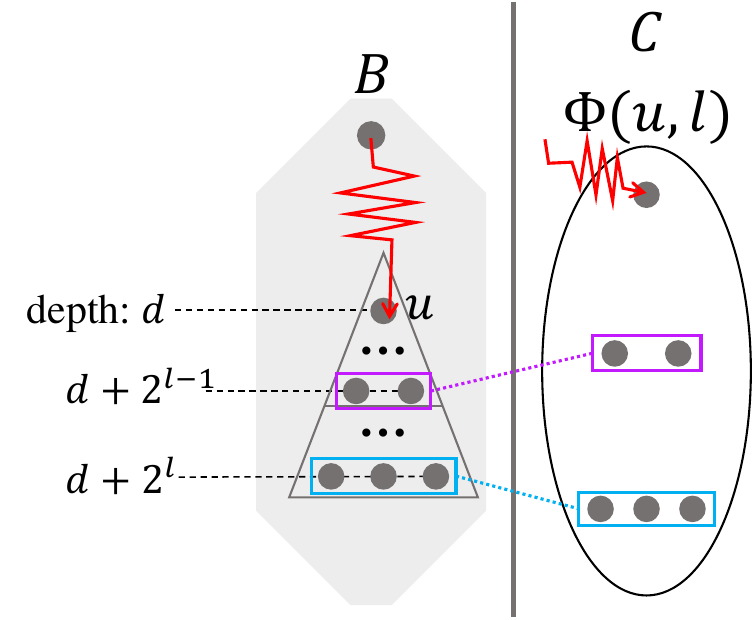}
    \caption{
\revision{
    (Color online.)
}
    A diagram illustrating \revision{the proof of} \Cref{Lem:NumberMidExit2}.}
    \label{Fig:NumberMidExit}
\end{figure}

\revision{ We next move on to bounding the number of vertices of $C$.\Cref{Lem:NumberMidExit1New} implies that the number of vertices per grouping is bounded by $\revision{\mathcal{O}}(|B|)$.
Using the $\revision{\mathcal{O}}(|B|\log|B|)$ bound on $|G(C)|$, we immediately get an $\revision{\mathcal{O}}(|B|^2\log|B|)$ bound on the total number of vertices. 
However, we can get a quadratic bound using a finer-grained counting argument based on the depths of the nodes in $B$.
For a given node $u \in B$, the following lemma bounds the number of (middle- and exit-)vertices of $\Phi(u, l)$ using $\mathcal{D}$:
}

\begin{lemma}\label{Lem:NumberMidExit2}
    Let $u$ be a node in $B$ such that $d\eqdef$\depth[u] is a multiple of $2 ^ l$ (i.e. $(u, l) \in \NS(B)$):
    \begin{enumerate}[label=(\alph*)]
        \item if $l \ge 1$, the number of middle vertices of $\Phi(u, l)$ is no more than $\mathcal{D}(d + 2 ^ {l-1})$
        \item the number of exit vertices of $\Phi(u, l)$ is no more than $\mathcal{D}(d + 2 ^ l)$
    \end{enumerate}
\end{lemma}

\begin{proof}
    \revision{This lemma follows directly from \Cref{Lem:NumberMidExit1New}} by considering the
\revision{
    depths of the nodes in $B$,
}
    as shown in \Cref{Fig:NumberMidExit}.
    \revision{By definition, $u \triangleright \Phi(u, l)$, so the number of exit vertices of $\Phi(u, l)$ equals the number of leaf nodes of BDDpatch($u, 2 ^ l$). Moreover, if $l \ge 1$, the number of middle vertices of $\Phi(u, l)$ equals the number of leaf nodes of BDDpatch($u, 2 ^ {l - 1}$).}
    
    \revision{Consequently,} the leaf nodes of BDDpatch$(u, 2^l)$ are a subset of (and possibly equal to) the nodes at depth $d + 2 ^ l$ in $B$, which establishes \textit{(b)}. \revision{If $l \ge 1$,} the leaf nodes \revision{of BDDpatch$(u, 2 ^ {l - 1})$} are a subset of (and possibly equal to) the nodes at depth $d + 2 ^ {l - 1}$ in $B$, which establishes \textit{(a)}.
\end{proof}

\revision{
Using \Cref{Lem:NumberMidExit2}, we can establish the quadratic bound: 
}

\begin{theorem}\label{Trm:Vertices}
    \revision{Let $B$ and $C$ be the BDD and CFLOBDD, respectively, for the same Boolean function $f$ (with the same variable ordering $Ord$), and} let $|V(C)|$ be the number of vertices in $C$. Then $|V(C)|=\revision{\mathcal{O}}(|B|^2)$.
\end{theorem}
\begin{proof}
    Now the ``size operator'' $(|\cdot|)$
    in \Cref{Eq:SizeCToBeReinterpreted}
    counts the number of vertices. According to \Cref{Lem:NumberMidExit2}, for $l \ge 1$, $\Phi(u, l)$ has at most $\mathcal{D}(\depth[u] + 2 ^ {l - 1})$ middle vertices, and $\mathcal{D}(\depth[u] + 2 ^ l)$ exit vertices. 
    
    While proving the $\revision{\mathcal{O}}(|B|^2)$ bound for the number of vertices, we will ignore those at level-0, as mentioned in \Cref{Se:TerminologyConventions}. Also, we will ignore the entry vertices, because each grouping has exactly one entry vertex and thus by \Cref{Trm:Groupings} the total number of entry vertices is $\revision{\mathcal{O}}(|B|\log|B|)$. 
    
    We obtain an asymptotic bound on the total number of vertices by putting the value $\mathcal{D}(\depth[u] + 2 ^ {l - 1}) + \mathcal{D}(\depth[u] + 2 ^ l)$ into \Cref{Fml:Overcount}:

    \begin{align*}
        |V(C)| &\le \sum_{l=0}^k  \sum_{d = 0}^{2^{k-l} - 1} \sum_{\substack{u \in B \\ \depth[u] = d \times 2^l}} |\Phi(u, l)| \\
        & \text{(pulling out level-0 and the entry vertices)}\\
        &\le \revision{\mathcal{O}}(|B|\log|B|) + \sum_{l=1}^k  \sum_{d = 0}^{2^{k-l} - 1} \sum_{\substack{u \in B \\ \depth[u] = d \times 2^l}} \mathcal{D}(d \times 2^l + 2 ^ {l - 1}) + \mathcal{D}(d \times 2^l + 2 ^ l) \\
        &= \revision{\mathcal{O}}(|B|\log|B|) + \sum_{l=1}^k  \sum_{d = 0}^{2^{k-l} - 1}  \mathcal{D}(d \times 2^l) \times (\mathcal{D}(d \times 2^l + 2 ^ {l - 1}) + \mathcal{D}(d \times 2^l + 2 ^ l)) \\
        &= \revision{\mathcal{O}}(|B|\log|B|) + \sum_{l=1}^k  \sum_{d = 0}^{2^{k-l} - 1}  \mathcal{D}(d \times 2^l) \times \mathcal{D}(d \times 2^l + 2 ^ {l - 1}) + \mathcal{D}(d \times 2^l) \times \mathcal{D}(d \times 2^l + 2 ^ l)
    \end{align*}

\begin{figure}
    \centering
    \includegraphics[width=1.0\linewidth]{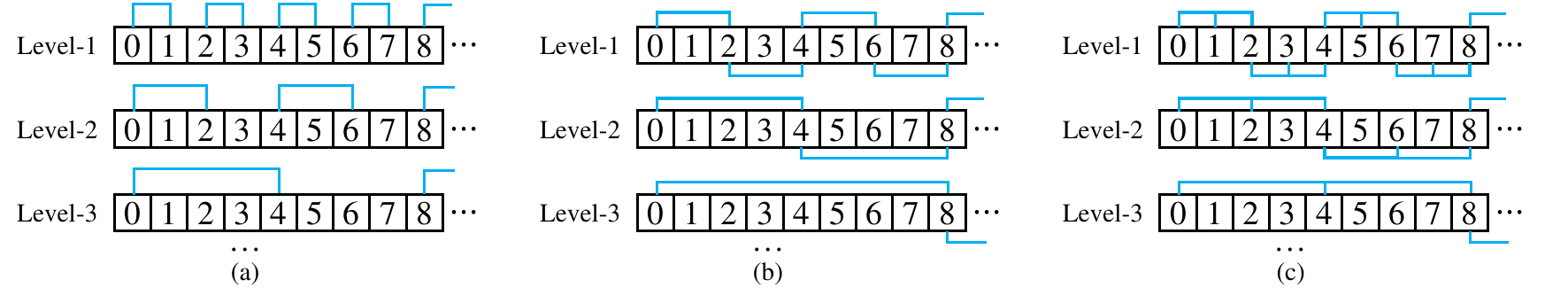}
    \caption{
      (a) and (b) The iteration patterns in the formula in \Cref{Trm:Vertices}.
      (c) The iteration pattern in the formula in \Cref{Trm:Edges}.
    }
    \label{Fig:IteratePattern}
\end{figure}

\noindent
In the double-summation term, for every pair ($i_0, i_i$) such that $0 \le i_0 < i_1 \le 2 ^ k$, the term $\mathcal{D}(i_0) \times \mathcal{D}(i_1)$ appears no more than twice.
The reason why this property holds is because on each successive iteration of the outer $\Sigma$, the stride $2^l$ is doubled,
and the inner $\Sigma$ iterates linearly with that stride. The scenarios for the terms $\mathcal{D}(d \times 2^l) \times \mathcal{D}(d \times 2^l + 2 ^ {l - 1})$ and $\mathcal{D}(d \times 2^l) \times \mathcal{D}(d \times 2^l + 2 ^ {l})$ are shown in \Cref{Fig:IteratePattern}(a) and \Cref{Fig:IteratePattern}(b), respectively. We overcount by assuming that for all such pairs $(i_0, i_1)$, $\mathcal{D}(i_0) \times \mathcal{D}(i_1)$ appears twice.
    
    \begin{align*}
        &\le \revision{\mathcal{O}}(|B|\log|B|) + 2 \times \sum_{0 \le i_0 < i_1 \le 2^k}  \mathcal{D}(i_0) \times \mathcal{D}(i_1) \\
        &\le \revision{\mathcal{O}}(|B|\log|B|) + \left(\sum_{i = 0} ^ {2 ^ k} \mathcal{D}(i)\right) ^ 2 \\
        &= \revision{\mathcal{O}}(|B|^2)
    \end{align*}

\end{proof}

\subsection{The Number of Edges}
\label{Subse: Count Edges}

We count the number of edges based on the following fact:

\begin{lemma}\label{Lem:NumEdges}
For a grouping $g$ with $M$ middle vertices and $E$ exit vertices, there are at most $M \times E$ B-return edges (i.e., the return edges of the B-connections).
\end{lemma}

\begin{proof}
    According to the structural invariants, the B-return-maps must be injective. Then each B-return-map will contain at most $E$ elements. The $M$ B-return-maps will contains at most $M \times E$ elements in total. That is, the number of B-return edges is no more than $M \times E$.
\end{proof}

Using a sum similar to the one used in \Cref{Trm:Vertices}, we can obtain the following result:
\begin{theorem}\label{Trm:Edges}
    \revision{Let $B$ and $C$ be the BDD and CFLOBDD, respectively, for the same Boolean function $f$ (with the same variable ordering $Ord$), and} let $|E(C)|$ be the number of edges in $C$. Then $|E(C)|=\revision{\mathcal{O}}(|B|^3)$.
\end{theorem}

\begin{proof}
    Now the ``size operator'' in \Cref{Eq:SizeCToBeReinterpreted} counts the edges of a grouping. According to \Cref{Lem:NumberMidExit2} and \Cref{Lem:NumEdges}, for $l \ge 1$, $\Phi(u, l)$ has at most $\mathcal{D}(\depth[u] + 2 ^ {l - 1}) \times \mathcal{D}(\depth[u] + 2 ^ l)$ B-return edges.

    While proving the $\revision{\mathcal{O}}(|B|^3)$ bound for the number of edges, we only need to consider the B-return edges---the number of call edges equals the total number of entry vertices and middle vertices, which contribute only $\revision{\mathcal{O}}(|B|^2)$; according to the structural invariants of CFLOBDDs, the total number of A-return edges equals the total number of middle vertices, which contribute only $\revision{\mathcal{O}}(|B|^2)$.
    Also, as mentioned in \Cref{Se:TerminologyConventions}, we will ignore the edges of groupings at level 0.

    We obtain an asymptotic bound on the total number of edges by putting the value $\mathcal{D}(\depth[u] + 2 ^ {l - 1}) \times \mathcal{D}(\depth[u] + 2 ^ l)$ into \Cref{Fml:Overcount}:
    \begin{align*}
        |E(C)| &\le \sum_{l=0}^k  \sum_{d = 0}^{2^{k-l} - 1} \sum_{\substack{u \in B \\ \depth[u] = d \times 2^l}} |\Phi(u, l)| \\
        & \text{(pulling out the call edges and the A-return edges) }\\
        &= \revision{\mathcal{O}}(|B|^2) + \sum_{l=1}^k  \sum_{d = 0}^{2^{k-l} - 1}  \mathcal{D}(d \times 2 ^ l) \times \mathcal{D}(d \times 2^l + 2 ^ {l - 1}) \times \mathcal{D}(d \times 2^l + 2 ^ l) \\
    \end{align*}

\noindent
In the double-summation term, for every triple ($i_0, i_1, i_2$) such that $0 \le i_0 < i_1 < i_2 \le 2 ^ k$ , the term $\mathcal{D}(i_0) \times \mathcal{D}(i_2) \times \mathcal{D}(i_2)$ appears no more than once. The reason why this property holds is because on each successive iteration of the outer $\Sigma$, the stride $2^l$ is doubled, and the inner $\Sigma$ iterates linearly by that stride. The scenario for the term $\mathcal{D}(d \times 2 ^ l) \times \mathcal{D}(d \times 2^l + 2 ^ {l - 1}) \times \mathcal{D}(d \times 2^l + 2 ^ l)$ is shown in \Cref{Fig:IteratePattern}(c). We overcount by assuming that for all triples $(i_0, i_1, i_2)$, $\mathcal{D}(i_0) \times \mathcal{D}(i_2) \times \mathcal{D}(i_2)$ appears once.

    \begin{align*}
        &\le \revision{\mathcal{O}}(|B|^2) + \sum_{0 \le i_0 < i_1 < i_2 \le 2 ^ k}\mathcal{D}(i_0) \times \mathcal{D}(i_1) \times \mathcal{D}(i_2) \\
        &\le \revision{\mathcal{O}}(|B|^2) + \frac{1}{3} \times \left(\sum_{i = 0} ^ {2 ^ k} \mathcal{D}(i)\right) ^ 3 \\
        &= \revision{\mathcal{O}}(|B|^3)
    \end{align*}
\end{proof}

%% file: 6_tight-instances.tex
\section{Tight Instances}
\label{Se:TightInstances}

This section present a series of functions that create tight instances for all three bounds. 

If we examine the proofs in \Cref{Se:Counting}, we might initially think that the upper bounds are very loose---we did much overcounting, especially by adding a lot of terms in the proof of the bounds. 

However, in constructing such worst-case instances, we are permitted to choose the order of the function's variables (the same order is used for a function's BDD and the function's CFLOBDDs). This section's results show that the three bounds of \Cref{Se:Counting} are asymptotically optimal. The first piece of intuition why the result holds stems from the following fact:
\begin{quote}
    The BDD nodes are not distributed evenly across depths: although $|B|=\sum_{i=0}^{2^k} \mathcal{D}(i)$, it is possible that for some depth $d$ (possibly more than one), $\mathcal{D}(d) = \Omega(|B|)$ (i.e., $\mathcal{D}(d)$ is proportional to $|B|$).
\end{quote}    
Consequently, if there exists a depth $d$ such that $\mathcal{D}(d)$ groupings occur at (almost) every level, the total number of groupings would be $\Omega(|B|\log|B|)$; if there are two such depths $d_0, d_1$ the number of vertices involves a term of the form $\mathcal{D}(d_0) \times \mathcal{D}(d_1)$, then the total number of vertices would be $\Omega(|B|^2)$; if there are three such depths $d_0, d_1, d_2$ and the number of edges involves a term of the form $\mathcal{D}(d_0) \times \mathcal{D}(d_1) \times \mathcal{D}(d_2)$, then the total number of edges would be $\Omega(|B| ^ 3)$.

Following this idea, we construct a series of Boolean functions as follows:
\begin{definition} \label{Def:TightInstanceFunc}
    Suppose that $k$ is big enough (or more precisely, $k \ge 7$). Let $f_k$ be a function $B ^ {2 ^ k} \to B$.
    
    $f_k(\vec{X}) \eqdef (\vec{Y_0} = \vec{Z_0} \land \vec{Y_1} = \vec{Z_1} \land \vec{Y_2} = \vec{Z_2})$, where:
    \begin{itemize}
        \item $\vec{Y_0} = \vec{X}[\frac{1}{2} \times 2 ^ k - k, \frac{1}{2} \times 2 ^ k ],\ \vec{Z_0} = \vec{X}[\frac{1}{2} \times 2 ^ k, \frac{1}{2} \times 2 ^ k + k]$
        \item $\vec{Y_1} = \vec{X}[\frac{5}{8} \times 2 ^ k - k, \frac{5}{8} \times 2 ^ k ],\ \vec{Z_1} = \vec{X}[\frac{5}{8} \times 2 ^ k, \frac{5}{8} \times 2 ^ k + k]$
        \item $\vec{Y_2} = \vec{X}[\frac{3}{4} \times 2 ^ k - k, \frac{3}{4} \times 2 ^ k ],\ \vec{Z_2} = \vec{X}[\frac{3}{4} \times 2 ^ k, \frac{3}{4} \times 2 ^ k + k]$
    \end{itemize}
\end{definition}
    
The shape of the BDD for the function $f_k$ ($B(f_k)$) is shown in \Cref{Fig: TightInstanceBDD}. The left column (depicted in blue and green) shows the paths that can lead to the value $T$, and the right column (depicted in yellow) shows the paths along which an inequality holds, and goes to the value $F$ without any further non-trivial decisions.  Notably, there are three denser regions in the left column, which performs the equality checks for $\vec{Y_0} = \vec{Z_0}$, $\vec{Y_1} = \vec{Z_1}$, and $\vec{Y_2} = \vec{Z_2}$, respectively. We can see from the partition of $\vec{X}$ that they lie exactly at ``1/2'', ``5/8'', and ``3/4'' of the total depths, respectively.

Let us look at the one for $\vec{Y_0}=\vec{Z_0}$ in detail, as shown in the top right of \Cref{Fig: TightInstanceBDD}. For an equality relation, the use of ``concatenated'' vocabularies
leads to a BDD whose size is exponential in the number of variables.
(However, our construction is allowed to use such a ``sub-optimal'' variable order.)
Over the $k$ plies that correspond to $\vec{Y_0}$, the diagram branches into $2 ^ k$ nodes to ``remember'' the $\vec{Y_0}$ values. Let us call the nodes at depth $2^{k-1}$ $n_0, n_1, \cdots, n_{2^k-1}$ from left to right. Each of these nodes ``remembers'' a certain value of $\vec{Y_0}$: $n_0$ remembers $0^k$; $n_1$ remembers $0^{k-1}1$; $\cdots$; $n_{2^k-1}$ remembers $1^k$. Then, over the $k$ plies that correspond to $\vec{Z_0}$, the diagram checks whether the value of $\vec{Z_0}$ matches the remembered value, directing each non-matching decision to the right column. The ``diamonds'' for $\vec{Y_1}=\vec{Z_1}$ and $\vec{Y_2}=\vec{Z_2}$ are similar, but each diamond has an additional node, e.g., $n'_{\textit{ineq}}$ in the case of the diamond for $\vec{Y_1}=\vec{Z_1}$, to receive the paths along which an inequality has already been discovered. If an assignment passes all three checking parts, then it will eventually reach the terminal $T$.

\begin{figure}
    \centering
    \includegraphics[width=1\linewidth]{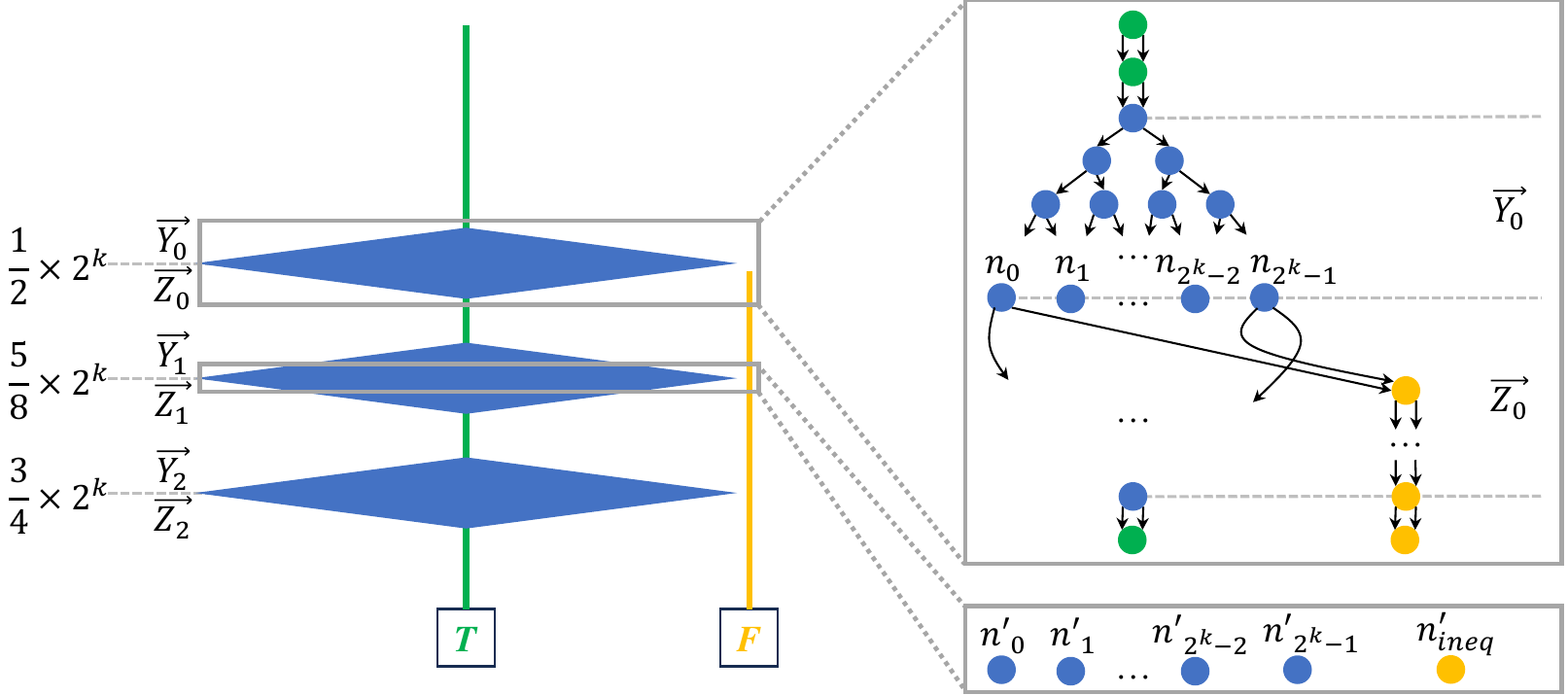}
    \caption{
\revision{
    (Color online.)
}
    A diagram illustrating the shape of the BDD for $f_k$. 
    }
    \label{Fig: TightInstanceBDD}
\end{figure}

It is easy to see that the number of nodes in $B(f_k)$ is $\Theta(2^k)$.\footnote{
  To state the number precisely, there are $\frac{21}{2} \times 2 ^ k - 6k - 8$ nodes in $B(f_k)$---there are two leaf nodes, $9 \times 2 ^ k - 6$ nodes in the three ``dense'' parts (blue), $2 ^ k - 6k - 3$ nodes in the green part, and $2 ^ {k - 1}$ nodes in the yellow part. 
}

Let's now consider the structure of $C(f_k)$.
Because the structure of $C(f_k)$ is far more complex than the structure of $B(f_k)$, we will not illustrate all parts of $C(f_k)$;
instead, we will focus only on the groupings that contribute predominantly to the size of $C(f_k)$, as depicted in \Cref{Fig: TightInstanceCFLOBDD}. Via the mapping $\Phi$, we will establish how certain groupings in $C(f_k)$ must be organized.

We first show that for ``almost all'' of the CFLOBDD levels $l$, $\Phi(n_i, l)$ is ``injective'':

\begin{lemma}\label{Lem: TightInstanceProperty}
For any $l$ such that $\lceil log_2 k\rceil \le l < k $, if $i \neq j$, then $\Phi(n_i, l) \neq \Phi(n_j, l)$.
\end{lemma}

\begin{proof}
    According to \S6 of \cite{cflobdd_arxiv},\footnote{
\revision{
    Reference \cite{cflobdd_arxiv} is an extended version of reference \cite{TOPLAS:SCR24} that contains proofs and other material that was omitted from \cite{TOPLAS:SCR24} due to length considerations.
    \S6 of \cite{cflobdd_arxiv} presents a denotational semantics for CFLOBDDs, which does not appear in reference \cite{TOPLAS:SCR24}.
}
}
    a level-$l$ proto-CFLOBDD represents a partition of $\{0, 1\} ^ {2 ^ l}$. We will show that the proto-CFLOBDDs headed by $\Phi(n_i, l)$ and $\Phi(n_j, l)$ represent different partitions of strings. 

    Consider the set of string $S = \{s\cdot 0^{2 ^ l -k} \mid s \in \{0, 1\} ^ k  \} \subseteq \{0, 1\} ^ {2 ^ l}$. From the standpoint of $\Phi(n_i, l)$ and $\Phi(n_j, l)$, this set enumerates all possible assignments to $\vec{Z_0}$ and constrains the other parts to be $0$. According to \Cref{Def:TightInstanceFunc} and the discussion on its BDD representation, the proto-CFLOBDD headed by $\Phi(n_i, l)$ partitions $S$ into two parts: it accepts the only string that matches the remembered value of $\vec{Y_0}$ and rejects the other strings. This property implies that $\Phi(n_i)$ and $\Phi(n_j)$ partition the strings in $S$ differently, and thus they must partition $\{0, 1\} ^ {2 ^ l}$ differently. Consequently, $\Phi$ must map $(n_i, l)$ and $(n_j, l)$ to different groupings.
\end{proof}

With \Cref{Lem: TightInstanceProperty}, we can establish that $f_k$ indeed defines a family of tight instances:

\begin{theorem}\label{Thm:TightBounds}
    For $f_k$, $|G(C)| = \Omega(k \times 2 ^ k) = \Omega(|B|\log|B|)$; $|V(C)| = \Omega(2 ^ {2k}) = \Omega(|B| ^ 2)$; $|E(C)| = \Omega(2 ^ {3k}) = \Omega(|B| ^ 3)$.
\end{theorem}
\begin{proof}
    The third term in each inequality follows from the second because $|B(f_k)| = \Theta(2 ^ k)$. Parts of the highest three levels of $C(f_k)$ are shown in \Cref{Fig: TightInstanceCFLOBDD}.

\begin{figure}
    \centering
    \includegraphics[width=1\linewidth]{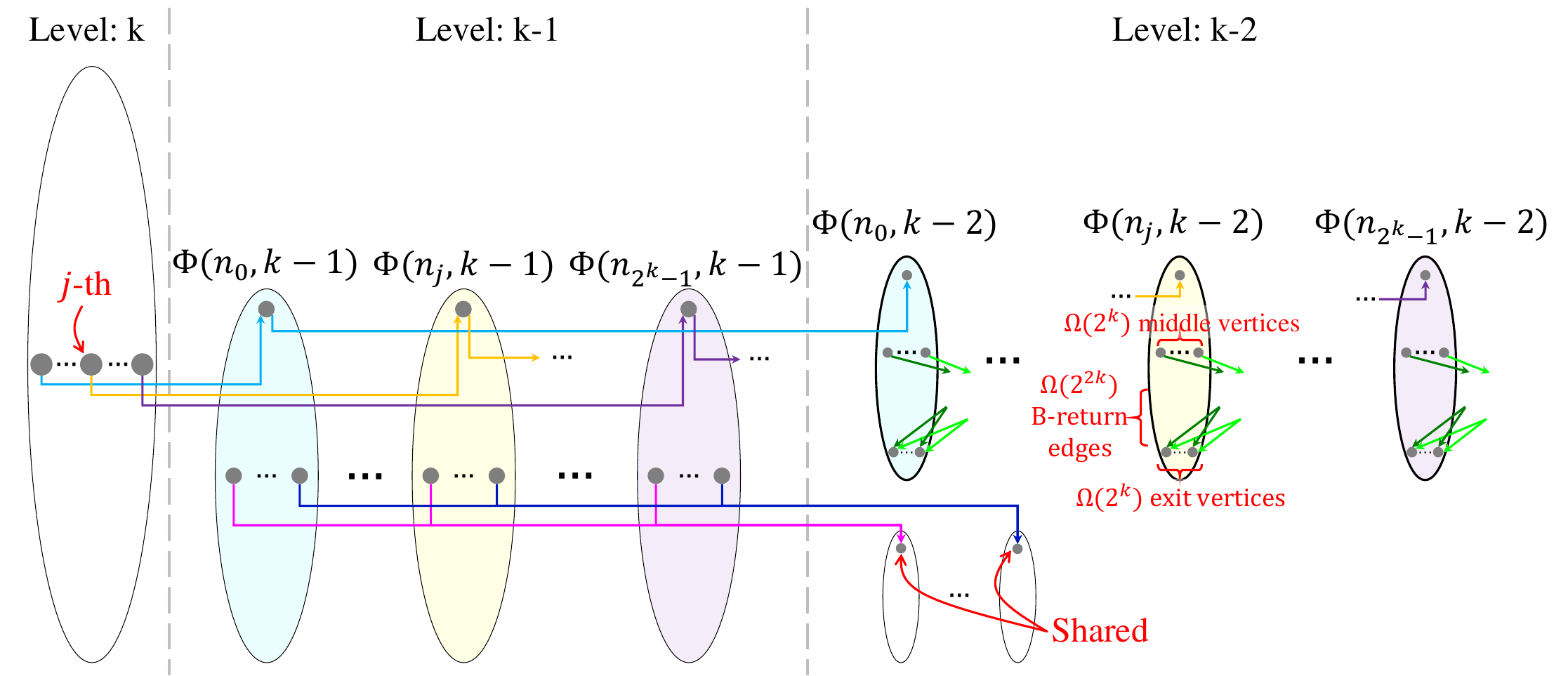}
    \caption{
\revision{
    (Color online.)
}
    A diagram illustrating the structure of the CFLOBDD for $f_k$. 
    To reduce cluster, we do not show some of the call edges and return edges that would link the level-($k\text{-}1$) groupings to level $k - 2$.
    As argued in \Cref{Thm:TightBounds}, the set of groupings $\{\Phi(n_i, k - 2) \mid 0 \le i < 2 ^ k\}$ contributes, in total, $\Omega(2 ^ {2k})$ vertices and $\Omega(2 ^ {3k})$ edges to $C$.
\revision{
    (To reduce clutter, return edges are elided in these diagrams.)
}
    }
    \label{Fig: TightInstanceCFLOBDD}
\end{figure}

    For groupings, we consider the groupings of $C(f_k)$ in $\{\Phi(n_i, l) \mid 0 \le i < 2 ^ k,\ \lceil log_2k \rceil \le l \le k - 1 \}$. According to \Cref{Lem: TightInstanceProperty}, these $\Phi(n_i, l)$s do not equal to each other. Thus, they will contribute $(k - \lceil log_2k \rceil) \times 2 ^ k = \Omega(k \times 2 ^ k)$ groupings.
    
    For vertices, we consider the vertices in the groupings $\{\Phi(n_i, k - 2) \mid 0 \le i < 2 ^ k\}$ in $C(f_k)$.
    According to \Cref{Lem:NumberMidExit2}, we can determine the number of vertices in these groupings by examining the corresponding BDDpatches. For any $0 \le i < 2 ^ k$, all the nodes in the lower half of $B(f_k)$ are reachable from $n_i$, indicating that BDDpatch($n_i, 2 ^ {k-2}$) includes all the nodes at depth $\frac{5}{8} \times 2 ^ k$ and $\frac{3}{4} \times 2 ^ k$. 
    The depths $\frac{5}{8} \times 2 ^ k$ and $\frac{3}{4} \times 2 ^ k$ are exactly the half-depth and the depth of the bottom of 
    BDDpatch($n_i, 2 ^ {k - 2}$)---see \Cref{Fig: TightInstanceBDD}---so 
    there are $2 \times 2 ^ k+ 3$ vertices (1 entry, $2 ^ k + 1$ middle, and $2 ^ k + 1$ exit vertices) in each of the groupings in $\{\Phi(n_i, k - 2) \mid 0 \le i < 2 ^ k\}$, as shown on the right in \Cref{Fig: TightInstanceCFLOBDD}.
    These $2 ^ k$ groupings contribute $2 ^ k \times (2 ^ k+ 3) = \Omega(2 ^ {2k})$ vertices in total to $C$.

    For edges, we also consider the edges in the $2^k$ groupings $\{\Phi(n_i, k - 2) \mid 0 \le i < 2 ^ k\}$ in $C(f_k)$. The previous paragraph just showed that each of these groupings contains $2 ^ k + 1$ middle vertices and $2 ^ k + 1$ exit vertices; we next show that each of these groupings contain $\Omega(2 ^ {2k})$ edges. Let us call the nodes at depth $\frac{5}{8} \times 2 ^ k$ $n'_0, n'_1, \cdots, n'_{2 ^ k - 1}, n'_{\textit{ineq}}$, as shown in the bottom-right of \Cref{Fig: TightInstanceBDD}.
    Consider an arbitrary node $n_j$ (at depth $\frac{1}{2} \times 2 ^ k$).
    The $2 ^ k + 1$ B-callees of $\Phi(n_j, k - 2)$ are exactly the set of groupings $\{\Phi(n'_i, k - 3) \mid 0 \le i < 2 ^ k\} \cup \{ \Phi(n'_{\textit{ineq}}, k - 3) \}$.\footnote{
      The proof of this observation requires us to ``unfold'' the definition of $\Phi$ and go back to an argument about paths. Let $\alpha$ be the partial assignment that lead us to $n_j$ in $B(f_k)$ and $\Phi(n_j, k - 2)$ in $C(f_k)$, respectively. The A-callee of $\Phi(n_j, k - 2)$ must be $\Phi(n_j, k - 3)$: let $g_A$ be the A-callee of $\Phi(n_j, k - 2)$; $n_j \triangleright g_A$ definitely holds by extending the path interpreting $\alpha$ in $C(f_k)$ that witnesses $n_j \triangleright \Phi(n_j, k - 2)$ one step (the step from the entry vertex of $\Phi(n_j, k - 2)$ to the entry vertex of $g_A$).

      At this point, we get two bijections: 
      \begin{enumerate}
        \item
          According to \Cref{Lem:Corr} (as well as the definition of $\Phi$), there exists a bijection between the set of leaf nodes of BDDpatch$(n_j, 2 ^ {k - 3})$ (which equals $\{n'_i \mid 0 \le i < 2 ^ k\} \cup \{ n'_{\textit{ineq}} \}$) and the exit vertices of $\Phi(n_j, k - 3)$, such that we always get to a corresponding leaf node/exit vertex pair for each extension of $\alpha$ with an assignment for the next $2 ^ {k - 3}$ variables in $B(f_k)$ and $C(f_k)$.
        \item
          The A-return map of $\Phi(n_j, k - 2)$ is an identity mapping. In other words, there is a bijection between the exit vertices of $\Phi(n_j, k - 3)$ and the middle vertices of $\Phi(n_j, k - 2)$.
      \end{enumerate}
      We can compose the two bijections and get a bijection between the nodes $\{n'_i \mid 0 \le i < 2 ^ k\} \cup \{ n'_{\textit{ineq}} \}$ and the middle vertices of $\Phi(n_j, k - 2)$ 
      such that we always get to a corresponding node/middle vertex pair for each extension of $\alpha$ with an assignment for the next $2 ^ {k - 3}$ variables in $B(f_k)$ and $C(f_k)$.

      Let $m$ be an arbitrary middle vertex of $\Phi(n_j, k - 2)$, and let $g_m$ be the B-callee associated with $m$.
      Suppose that $n'_m \in \{n'_i \mid 0 \le i < 2 ^ k\} \cup \{ n'_{\textit{ineq}} \}$
      is the node that corresponds to $m$ in the composed bijection.
      Any path/assignment that extends $\alpha$ from $n_j$ to $n'_m$ can be further extended in the same way we did for the A-connection part of $\Phi(n_j, k - 2)$, and thus $n'_m \triangleright g_m$.
      In other words, $g_m = \Phi(n'_m, k - 3)$. 
      There are a total of $2 ^ k + 1$ middle vertices of $\Phi(n_j, k - 2)$.
      The procedure just described for finding $\Phi(n'_m, k - 3)$ from $m$ establishes that 
      the B-callees of the $2 ^ k + 1$ middle vertices of $\Phi(n_j, k - 2)$ are exactly the set of groupings $\{\Phi(n'_i, k - 3) \mid 0 \le i < 2 ^ k\} \cup \{ \Phi(n'_{\textit{ineq}}, k - 3) \}$.
    }
    For the same reason that each grouping $\Phi(n_i, k - 2)$ has $2 ^ k + 1$ exit vertices, the grouping $\Phi(n'_i, k - 3)$ has $2 ^ k + 1$ exit vertices.
    Because the number of return edges in a B-connection equals the number of exit vertices of its B-callee, each of the $2 ^ k$ B-connections that call one of the groupings in the set $\{\Phi(n'_i, k - 3) \mid 0 \le i < 2 ^ k\}$ contains $2 ^ k + 1$ return edges, shown as the bright green and deep green return edges on the right in \Cref{Fig: TightInstanceCFLOBDD}, and thus each grouping $\Phi(n_j, k - 2)$ has $\Omega(2 ^ {2k})$ edges.
    Consequently, the groupings $\{\Phi(n_i, k - 2) \mid 0 \le i < 2 ^ k\}$ contribute a total of $\Omega(2 ^ {3k})$ edges to $C(f_k)$.
\end{proof}

%% file: 7_related.tex
\section{Related Work}
\label{Se:RelatedWork}

\revision{
A detailed discussion of work related to CFLOBDDs can be found in \cite[\S11]{TOPLAS:SCR24}.
In this section, we focus on work that relates most closely to the issue of data-structure sizes compared to BDD sizes.
}

\paragraph{\revision{Sentential Decision Diagrams and Variable Shift Sentential Decision Diagrams}}
\revision{
Other data structures that generalize BDDs are Sentential Decision Diagrams (SDDs)~\cite{SDD} and Variable Shift SDDs (VS-SDDs) \cite{VSSDD} (and their quantitative generalizations, such as Probabilistic SDDs~\cite{kisa2014probabilistic}).
As mantioned in \Cref{Se:Introduction}, these decision diagrams also use a tree-structured decomposition of the Boolean variables.
There are families of functions for which these data structures are double-exponentially succinct compared to decision trees and exponentially succinct compared to BDDs.
In SDDs and VS-SDDs, the variable decomposition to use is a user-supplied input.
}

\revision{
VS-SDDs support a sub-structure sharing property that goes beyond SDDs; however, as with BDDs and SDDs, the objects that are shared are always
sub-DAGs, which represent a Boolean function (i.e., $\{0, 1\} ^ n \to \{0, 1 \}$). In contrast, the CFLOBDD ``procedure-call'' device is a more flexible mechanism for reusing substructures than the mechanism for reusing substructures in VS-SDDs.
In particular, the call-return structure used in CFLOBDDs allows sharing of the ``middle of a DAG''---(i.e., in our terminology, a BDDpatch (\Cref{De:BDDpatch})---in the form of a shared proto-CFLOBDD.
In particular, a grouping $g$ can have multiple middle vertices that reuse the same B-connection proto-CFLOBDD $b$, as long as the return edges for the different invocations of $b$ use different mappings to $g$'s exit vertices.
This ``contextual rewiring'' gives CFLOBDDs greater ability to reuse substructures than SDDs and VS-SDDs.
(Moreover, $b$ can also be used as the A-connection proto-CFLOBDD of $g$.)

However, there are interesting connections---in both directions---between SDDs and VS-SDDs on the one hand and CFLOBDDs on the other.
For instance, SDDs and VS-SDDs \emph{can} represent something similar to BDDpatches.
Whereas a CFLOBDD represents the first-half of the variables with a proto-CFLOBDD whose outermost grouping has $k$ exit vertices, an SDD or VS-SDD can represent them with $k$ different sub-diagrams (which represent different Boolean-valued functions).
Each sub-diagram \emph{captures one partition set} of the assignments to the first-half variables.
Thus, from this perspective the major difference is that a CFLOBDD can have multi-exit groupings (groupings with more than two exits), whereas an SDD or VS-SDD would have many sub-diagrams, each having only one or two exits.

In the other direction, variants of CFLOBDDs in which the variable decomposition is a user-supplied input are also possible \cite[\S12]{cflobdd_arxiv}.
Consider a Boolean-valued function $f : \{0,1\}^{2^k} \rightarrow \{0, 1\}$ over variables $x_0, \dots , x_{{2^k}-1}$.\footnote{
\revision{
  For a Boolean function for which the number of variables is not a power of 2, one can pad the function with dummy variables to reach the next higher power of 2.
  Depending on the function, one might choose to interleave the dummy variables among the ``legitimate'' variables or place them at the end (or some combination of both).
}
}
In the variable decomposition used in CFLOBDDs, as one descends by one level, the variables are divided into two halves: $x_0, \dots, x_{{2^{k-1}}-1}$ for the A-connection and $x_{2^{k-1}},\dots, x_{{2^k}-1}$ for the B-connections.
This divide-variables-in-half decomposition is carried out recursively for the proto-CFLOBDDs in the A-connection and B-connections.
Without being precise about the indexing of the Boolean variables themselves, we can say that the variable decomposition used in CFLOBDDs---as presented in \Cref{Se:CFLOBDDs} and in \cite{TOPLAS:SCR24}---has the following structure, which defines a balanced binary tree:
\[
  \begin{array}{rcll}
  \Matched_k & \rightarrow & \Matched_{k-1} \quad \Matched_{k-1} & \\
             & \vdots &  & \\
  \Matched_0 & \rightarrow & x_j & \text{where $0 \leq j \leq 2^{k}-1$.}
  \end{array}
\]
}

\revision{
However, other decomposition schemes could be used.
As long as there is a fixed structural-decomposition pattern to how variables are divided as one descends to a grouping's ``callees,'' it is possible to keep the same properties that one has with ``standard'' CFLOBDDs---i.e., support for the operations of a BDD-like API, without having to instantiate the corresponding decision tree; canonicity; etc.
For instance, the three-vocabulary addition relation
$\ADD_n : \{0,1\}^{n} \times \{0,1\}^{n} \times \{0,1\}^{n} \rightarrow \{0,1\}$
on variables $\{ x_0\cdots x_{n-1} \}$, $\{ y_0\cdots y_{n-1} \}$, and $\{ z_0\cdots z_{n-1} \}$, defined by $\ADD_n(X,Y,Z) \eqdef Z = (X + Y \mod 2^{n})$ was considered by \citet[\S8.3]{cflobdd_arxiv}.
For $\ADD_n$, it would have been convenient to use the decomposition
\[
  \begin{array}{rcll}
    \Matched_k & \rightarrow & \Matched_{k-1} \quad \Matched_{k-1} \quad \Matched_{k-1} & \\
               & \vdots & & \\
    \Matched_0 & \rightarrow & x_j & \text{where $0 \leq j \leq 2^{k}-1$,}
  \end{array}
\]
which would have avoided having to ``waste'' an additional $n$ variables in a fourth vocabulary $\{ \wasted_0 \cdots \wasted_{n-1} \}$,
as in \citet[Figures 16 and 17]{cflobdd_arxiv}.
Another possibility would be to permit a level-$i$ grouping to have connections to level-$(i\text{-}j)$ groupings, where $j > 1$).
For instance, one could have a Fibonacci-like decomposition
\[
  \begin{array}{rcll}
    \Matched_k & \rightarrow & \Matched_{k-1} \quad \Matched_{k-2} & \\
               & \vdots & & \\
    \Matched_1 & \rightarrow &  x_j & \text{where $0 \leq j \leq 2^{k}-1$.}\\
    \Matched_0 & \rightarrow &  x_j & \text{where $0 \leq j \leq 2^{k}-1$.}
  \end{array}
\]
If one were to use the linear grammar
\begin{equation}
  \label{Eq:LinearGrammar}
  \begin{array}{rcll}
    \Matched_k & \rightarrow & x_j \quad \Matched_{k-1} & \text{where $0 \leq j \leq 2^{k}-1$} \\
               & \vdots & & \\
    \Matched_0 & \rightarrow & x_j & \text{where $0 \leq j \leq 2^{k}-1$.}
  \end{array}
\end{equation}
the variable decomposition is essentially a linked list of the Boolean variables.
The resulting CFLOBDD variant would be essentially isomorphic to BDDs, and the size of each CFLOBDD would be linearly related to that of the corresponding BDD.
}

\revision{
The observations that have been made about the sizes of SDDs and VS-SDDs compared to BDDs only consider a decomposition like \Cref{Eq:LinearGrammar}.
However, in that situation the linear-size relationship is obtained only by relinquishing a tree-like decomposition in favor of a list-like decomposition---but that throws the baby out with the bathwater: it is only by using a tree-like decomposition that families of SDDs and VS-SDDs, along with families of CFLOBDDs, obtain a best-case exponential-succinctness advantage over the corresponding family of BDDs.
In contrast, the cubic-size worse-case bound of \Cref{Trm:Edges} shows that standard CFLOBDDs (with its balanced binary-tree decomposition) are only polynomially larger than BDDs, while simultaneously retaining the property of enjoying an exponential-succinctness advantage over BDDs in the best case.
(Moreover, it would be relatively straightforward to establish that all of the CFLOBDD variants sketched above are at most polynomially larger than BDDs, and retain an exponential-succinctness advantage over BDDs for the ones based on a non-linear grammar.)
}

\revision{
For SDDs and VS-SDDs, it has not been explored whether there is a tree-structured decomposition that offers a polynomial worst-case bound, while simultaneously enjoying an exponential-succinctness advantage over BDDs in the best case.
We believe that the concepts presented in our paper, particularly BDDpatches (\Cref{De:BDDpatch}) and node stratification (\Cref{Def:NS}), are likely to be useful in resolving this question.
}

\paragraph{\revision{Other Kinds of Hierarchical Structure in Decision Diagrams}}
\revision{
There are other kinds of hierarchical organization that have been used in decision-diagram data structures.
Linear Inductive Functions (LIFs) \cite{LIFs} and Exponentially Inductive Functions (EIFs) \cite{Thesis:Gupta94} are BDD-based representations that use numeric or arithmetic annotations on the edges to represent functions over Boolean arguments, with BDD-like subgraphs connected in highly restricted ways.
In contrast, in CFLOBDDs, different groupings at the same level (or different levels) can have very different kinds of connections in them.
}

\revision{
Set Decision Diagrams \cite{couvreur2005hierarchical,thierry2009hierarchical}---which we will abbreviate as Set-DDs to distinguish them from the aforementioned Sentential Decision Diagrams---are based on Data Decision Diagrams (DDDs) \cite{DDD}, which are similar to trie data structures in that they store ``words'' that represent finite sets of assignments sequences of the form $(e_1 := x_1) \cdot (e_2 := x_2) \cdots (e_n := x_n)$ where $e_i$ are variables and $x_i$ are the assigned values.
Whereas tries share common prefixes of sequences in the trie, DDDs are DAG-structured, sharing common suffixes as well as common prefixes.
Set-DDs are similar, but represent sequences of assignments of the form $e_1 \in a_1$; $e_2 \in a_2$; $\ldots$ $e_n \in a_n$, where the $e_i$ are variables and the $a_i$ are sets of values.
}

\revision{
The notion of hierarchy used in Set-DDs is that they need a data structure to use for the sets that label edges of a Set-DD that needs to support an API that includes the usual set-theoretic operations (union, intersection, and set difference) and the computation of a hash key for the set represented.
While they provide an interface so that many types of decision diagrams can serve this purpose, they point out that because a Set-DD itself represents a set, a hierarchy can be introduced in Set-DDs by using Set-DDs as the needed set representation.
The structure is non-recursive: one has a Set-DD, whose labels are Set-DDs, whose labels are Set-DDs, and so on to some degree of nesting.
}

\revision{
This notion of hierarchy is orthogonal to the hierarchical structure of groupings in CFLOBDDs.
In CFLOBDDs, the hierarchical structure implements a kind of procedure-call mechanism.
As already mentioned above, such ``procedure calls'' allow there to be additional sharing of structure beyond what is possible in BDDs and many other kinds of decision diagrams:
a BDD can share only sub-DAGs, whereas a procedure call in a CFLOBDD shares the ``middle of a DAG'' in the form of a shared proto-CFLOBDD.
To convince oneself that the notions are truly orthogonal
\begin{itemize}
  \item
    CFLOBDDs, or Weighted CFLOBDDs (WCFLOBDDs) \cite{DBLP:journals/corr/abs-2305-13610}, could be used as the set representation in Set-DDs.
  \item 
    The kind of hierarchy found in Set-DDs could also be used in WCFLOBDDs, in which the edges of level-0 groupings are labeled with weights.
    Certain kinds of weights (e.g., function-valued or set-valued weights) could themselves be represented with WCFLOBDDs, whose level-0 groupings are in turn labeled with WCFLOBDDs, and so on to some degree of nesting.
\end{itemize}
These observations make it less likely that the concepts introduced in this paper (BDDpatches, node stratification, etc.), would be useful in analyzing how the sizes of Set-DDs are related to the sizes of other data structures.
}

%% file: 8_conclusion.tex
\section{Conclusion}
\label{Se:Conclusion}

BDDs are commonly used to represent Boolean functions;
CFLOBDDs serve as a plug-compatible alternative to them.
\revision{
While \citeauthor{TOPLAS:SCR24} established that there are \emph{best-case families of functions}, which demonstrate an inherently exponential-succinctness advantage of CFLOBDDs over BDDs
}
(regardless of what variable ordering is used in the BDD), no relationship in the opposite direction was known.
This paper fills in the picture by establishing that
\begin{itemize}
  \item
    For every BDD, the size of the corresponding CFLOBDD is at most a polynomial function of the BDD's size:
    \begin{quote}
      If BDD $B$ for function $f$ is of size $|B|$ and uses variable ordering $\textit{Ord}$, then the size of the CFLOBDD $C$ for $f$ that also uses $\textit{Ord}$ is bounded by $\revision{\mathcal{O}}(|B|^3)$.
    \end{quote}
  \item
    There is a family of functions for which $|C|$ grows as $\Omega(|B|^3)$, and hence the bound is tight. 
\end{itemize}

In \Cref{Se:Relations}, we established structural and semantic relationships between $B$ and $C$.
We first defined the binary relation $\triangleright$ (\Cref{Def:Simu}),
which relates two structural notions (paths in BDDs versus paths with a matching condition in CFLOBDDs) to a semantic notion (a partial assignment of a given length). We then introduced the $\NS$ relation (\Cref{Def:NS}); the mapping $\Phi$ (\Cref{def:phi}), and proved that $\Phi$ is an onto-mapping from $B$ to $C$ (\Cref{trm:phi}), which gave us important insights into the structure of a CFLOBDD when compared with the BDD for the same function, with the same variable ordering.

The formalization in \Cref{Se:Relations} has applications that go beyond the polynomial bounds of CFLOBDDs versus BDDs.
For example, by defining $\triangleright$, \NS, and $\Phi$ in the same way for Weighted CFLOBDDs (WCFLOBDDs) \cite{DBLP:journals/corr/abs-2305-13610} and Weighted BDDs (WBDDs) \cite{DBLP:journals/tcad/NiemannWMTD16,DBLP:conf/date/ViamontesMH04}, we can prove that the size of a WCFLOBDD is at most cubic in the size of the WBDD for the same function.
Additionally, $\Phi$ and \Cref{Lem:Corr} enable us to design an algorithm that 
converts a BDD to its equivalent CFLOBDD in polynomial time.

Based on these structural relationships, \Cref{Se:Counting} established an $\revision{\mathcal{O}}(|B|\log|B|)$ bound on the number of CFLOBDD groupings (\Cref{Trm:Groupings}), an $\revision{\mathcal{O}}(|B|^2)$ bound on the number of CFLOBDD vertices (\Cref{Trm:Vertices}), and an $\revision{\mathcal{O}}(|B|^3)$ bounds on the number of CFLOBDD edges (\Cref{Trm:Edges}).
In the proofs, we employed various mathematical techniques that involved significant overcounting.
Despite this overcounting, we demonstrated in \Cref{Se:TightInstances} that the bounds are tight (i.e., asymptotically optimal) by giving a series of functions $f_k$ (\Cref{Def:TightInstanceFunc}) that create tight instances for all three bounds.

\revision{
    We can now answer the question proposed in \Cref{Subse: 3/4-depth-duplication} with a more detailed explanation:
\begin{quote}
    ``3/4-depth duplication'' does not propagate across levels.
    \Cref{Cor:GroupingsAtLSmallerThanB} proves that there is a kind of ``level locality''---the number of ``essentially-different'' groupings that can arise for each level is bounded by $|B|$, which is fundamental to the polynomial bounds established in \Cref{Se:Counting}. 
 \end{quote}
     We can also explain the ``level-locality'' concept in a more operational way, from the viewpoint of sharing.
}
    The definition of $\Phi$ implies the presence of a form of necessary sharing: the grouping that simulates some BDDpatch$(u, 2 ^ l)$ (i.e, $\Phi(u, l)$) only needs to be constructed once, even if it is reached along different paths from the outermost grouping.
    For example, if we look at the CFLOBDD for $f_k$ (defined in \Cref{Se:TightInstances}), we can see that ``3/4-depth duplication'' indeed happens when going from level $k$ to level $k \text{-} 1$, but it would not propagate to level $k \text{-} 2$---even if each of the level-$(k\text{-}1)$ groupings has $2 ^ k + 1$ B-callees, these groupings will greatly share their B-callees, like $\Phi(n_0, k - 1)$ and $\Phi(n_{2 ^ k - 1}, k - 1)$ in \Cref{Fig: TightInstanceCFLOBDD}. As a result, we only need a total of $\revision{\mathcal{O}}(2 ^ k) = \revision{\mathcal{O}}(|B|)$ groupings for level $k - 2$.
    In that sense, duplication does not propagate. The initial intuition that ``3/4-depth duplication'' would propagate was a red herring: it incorrectly assumed that the groupings at level $k$ to level 0 are independent from each other, which neglects the form of sharing that actually occurs (as described above).

\paragraph{\revision{Future Work}}
There are \revision{several additional} questions that would be interesting to study regarding the \revision{sizes of CFLOBDDs and other decision diagrams} compared to BDDs.

\revision{
Following the methods developed in this paper,
it would be relatively straightforward to establish that the sizes of all CFLOBDD variants of the kind described in \Cref{Se:RelatedWork} are at most polynomially larger than BDDs:
we can always establish ``level locality'' by defining node stratification and a mapping $\Phi$ similar to \Cref{Def:NS} and \Cref{def:phi}, respectively.
The only adjustment required would be to modify the ``depth-alignment'' in the definition of node stratification, which, for the balanced binary-tree decomposition of standard CFLOBDDs, is based on multiples of $2^l$ (motivated by wanting natural strata that ``synchronize'' with CFLOBDD groupings at level $l$).
}

This paper demonstrates that if a BDD $B$ and a CFLOBDD $C$ for function $f$ use the same variable ordering $\textit{Ord}$, then $|C| = \revision{\mathcal{O}}(|B| ^ 3)$. However, what if we allow the variable ordering to change? What is the bound on the size of CFLOBDDs compared to BDDs when the CFLOBDD is permitted to use a different variable ordering? Could we achieve a better bound, such as $\revision{\mathcal{O}}(|B|^2)$, $\revision{\mathcal{O}}(|B|)$, or even sub-linear bounds? This problem remains open.

\revision{
Finally, we believe that the concepts presented in our paper, particularly BDDpatches (\Cref{De:BDDpatch}) and node stratification (\Cref{Def:NS}), can be useful for analyzing the complexity of other structures, such as SDDs \cite{SDD} and VS-SDDs \cite{VSSDD}.
}

%% file: 9_appendix-defCFLOBDD.tex
\appendix 
\section{The Structural Invariants of CFLOBDDs}
\label{Appendix: StructuralInvariants}

As described in \S2.2, the structure of a mock-CFLOBDD consists of different groupings organized into levels, which are connected by edges in a particular fashion.
In this section, we describe some additional \emph{structural invariants} that are imposed on CFLOBDDs, which go beyond the basic hierarchical structure that is provided by the
entry vertex, middle vertices, call edges, return edges, and exit vertices of a grouping.

Most of the structural invariants concern the organization of what we call \emph{return tuples}.
For a given $A$-call edge or $B$-call edge $c$ from grouping $g_i$ to $g_{i-1}$, the return tuple $rt_c$ associated with $c$ consists of the sequence of targets of return edges from $g_{i-1}$ to $g_i$ that correspond to $c$ (listed in the order in which the corresponding exit vertices occur in $g_{i-1}$).
Similarly, the sequence of targets of value edges that emanate from the exit vertices of the highest-level grouping $g$ (listed in the order in which the corresponding exit vertices occur in $g$) is called the CFLOBDD's \emph{value tuple}.

Return tuples represent mapping functions that map exit vertices at one level to middle vertices or exit vertices at the next greater level.
Similarly, value tuples represent mapping functions that map exit vertices of the highest-level grouping to terminal values.
In both cases, the $i^{\textit{th}}$ entry of the tuple indicates the element that the $i^{\textit{th}}$ exit vertex is mapped to.
Because the middle vertices and exit vertices of a grouping are each arranged in some fixed known order, and hence can be stored in an array, it is often convenient to assume that each element of a return tuple is simply an index into such an array.
For example, in \Cref{fig:CFLOBDDexample},
\begin{itemize}
  \item
    The return tuple associated with the A-call edge of $g_2$ is [1, 2].\footnote{
      Following the conventions used in \cite{TOPLAS:SCR24}, indices of array elements start at 1.
    }
  \item
    The return tuple associated with the B-call edge that starts at $v_4$ of $g_2$ is [1].
  \item
    The return tuple associated with the B-call edge that starts at $v_5$ of $g_2$ is [2, 1].
  \item
    The value tuple associated with the CFLOBDD is the 2-tuple [5, 7].
\end{itemize}

\subsubsection*{Rationale}
The structural invariants are designed to ensure that---for a given order on the Boolean variables---each Boolean function has a unique, canonical representation as a CFLOBDD \cite[Theorem 4.3]{TOPLAS:SCR24}.
In reading \Cref{De:CFLOBDD} below, it will help to keep in mind that the goal of the invariants is to force there to be a \emph{unique} way to fold a given decision tree into a CFLOBDD that represents the same Boolean function, which is discussed in \cite[\S4.2 and Appendix C]{TOPLAS:SCR24}. The main characteristic of the folding method is that it works greedily, left to right.
This directional bias shows up in structural invariants~\ref{Inv:1}, \ref{Inv:2a}, and \ref{Inv:2b}.

\begin{figure*}
\centering
\begin{tabular}{c@{\hspace{.33in}}c}
  \includegraphics[height=1.62in]{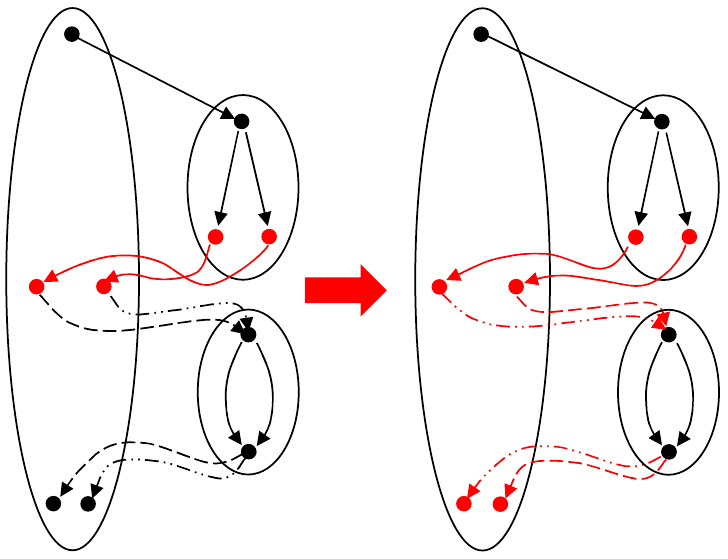}
  &
  \includegraphics[height=1.62in]{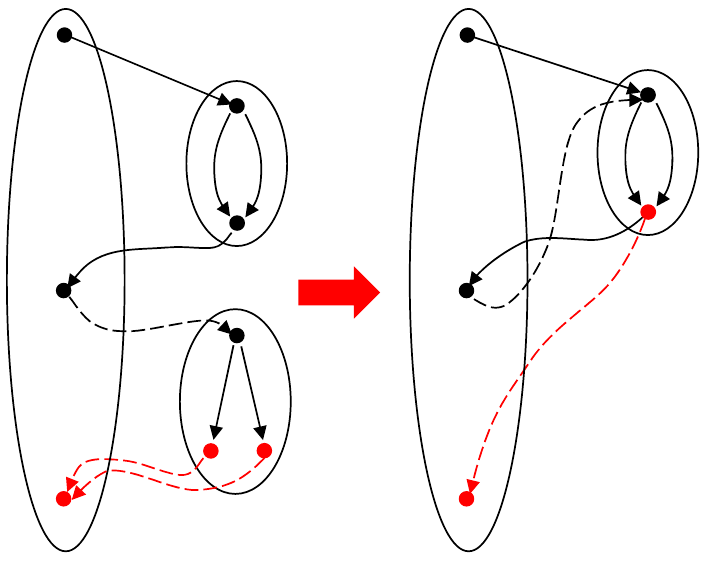}
  \\
  (a) Structural invariant~\ref{Inv:1}
  &
  (b) Structural invariant~\ref{Inv:2a}
  \\
  \includegraphics[height=1.62in]{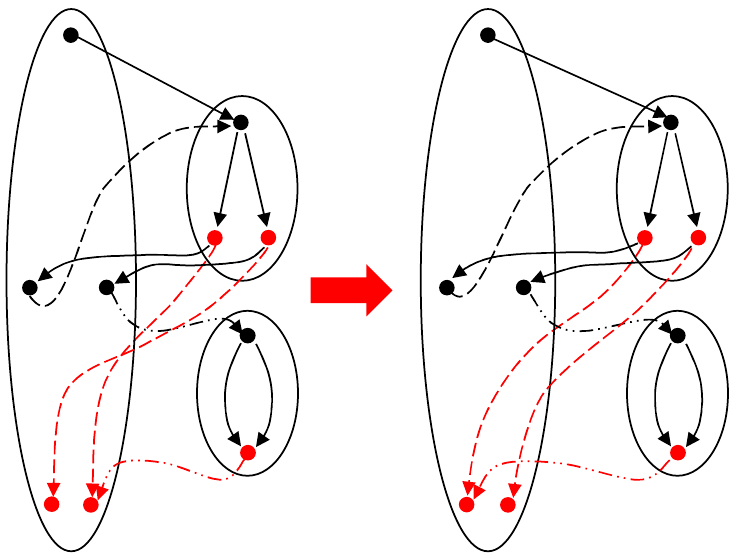}
  &
  \includegraphics[height=1.62in]{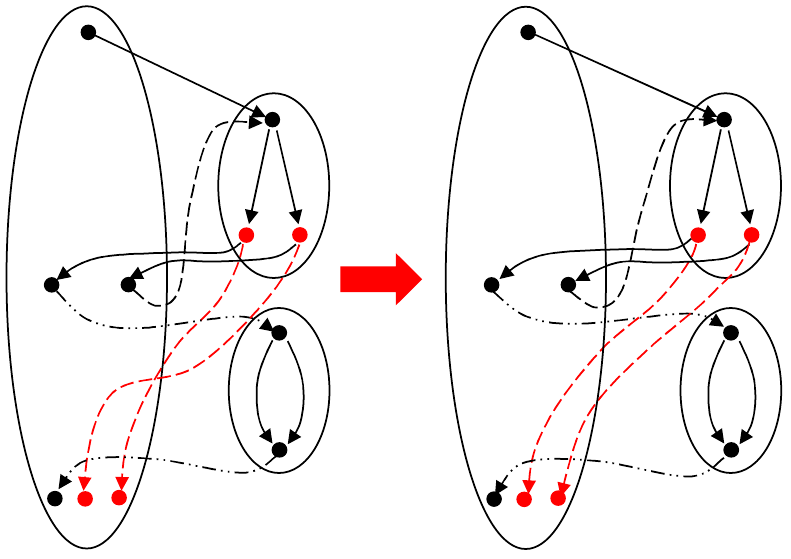}
  \\
  (c) Structural invariant~\ref{Inv:2b}
  &
  (d) Another case of structural invariant~\ref{Inv:2b}
  \\
  \includegraphics[height=1.62in]{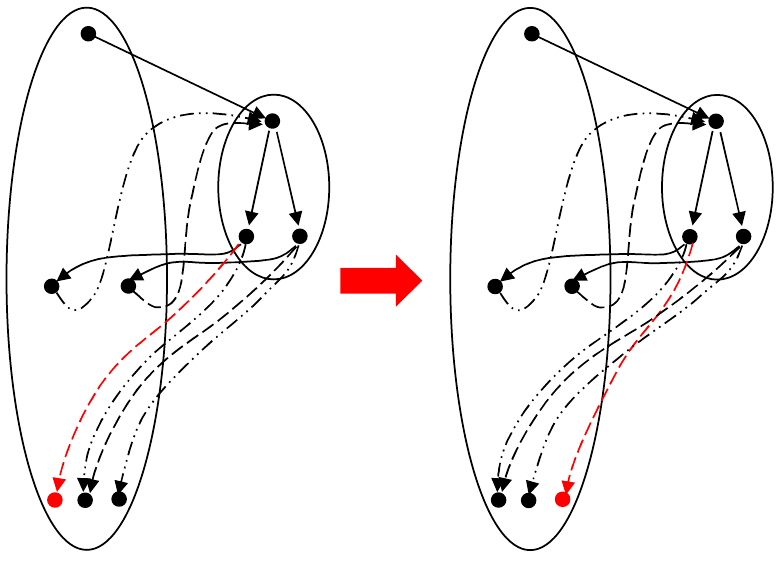}
  &
  \includegraphics[height=1.62in]{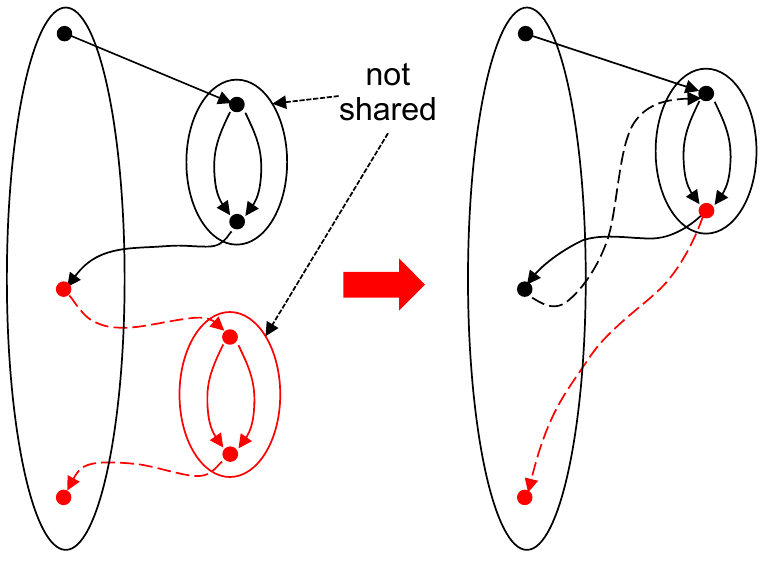}
  \\
  (e) Another case of structural invariant~\ref{Inv:2b}
  &
  (f) Structural invariant~\ref{Inv:3}
  \\
  \includegraphics[height=1.62in]{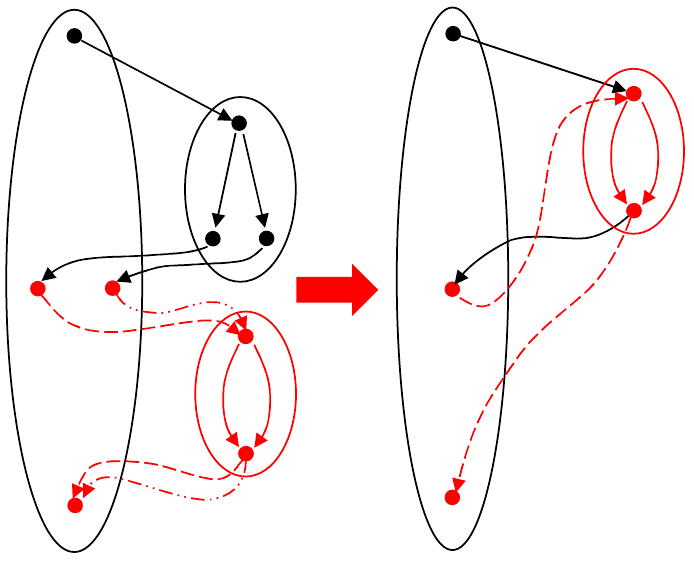}
  &
  \includegraphics[height=1.75in]{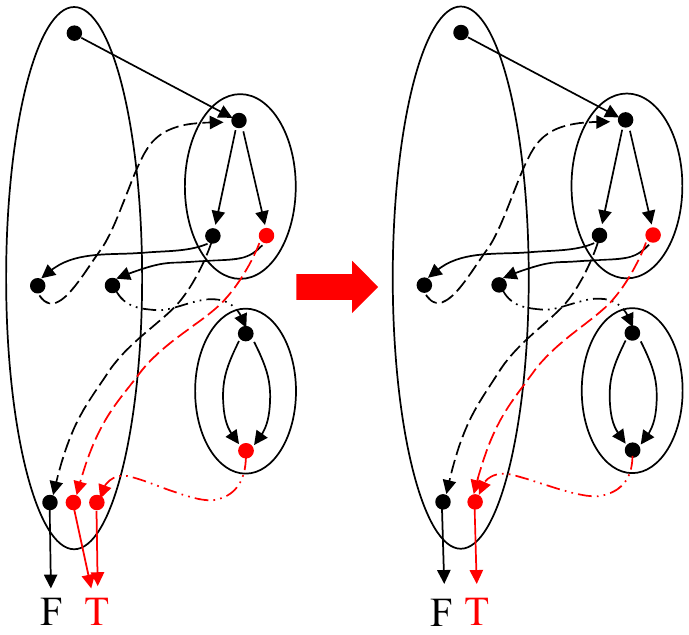}
  \\
  (g) Structural invariant~\ref{Inv:4}
  &  
  (h) Structural invariant~\ref{Inv:6}
\end{tabular}
\caption{\protect \raggedright
\revision{
    (Color online.)
}
To the left of each arrow, a mock-proto-CFLOBDD that violates the indicated structural invariant; to the right, a corrected proto-CFLOBDD.
Invariant violations and their rectifications are shown in red. 
}
\label{Fi:StructuralInvariantsIllustrated}
\end{figure*}

We can now complete the formal definition of a CFLOBDD.

\begin{definition}[Proto-CFLOBDD and CFLOBDD]\label{De:CFLOBDD}
    A \emph{proto-CFLOBDD} $n$ is a mock-proto-CFLOBDD in which every grouping/proto-CFLOBDD in $n$ satisfies the \emph{structural invariants} given below.  
    In particular, let $c$ be an $A$-call edge or $B$-call edge from grouping $g_i$ to $g_{i-1}$, with associated return tuple $rt_c$.
    \begin{enumerate}
      \item
        \label{Inv:1}
        If $c$ is an $A$-call edge, then $rt_c$ must map the exit vertices of $g_{i-1}$ one-to-one, and in order, onto the middle vertices of $g_i$:
        Given that $g_{i-1}$ has $k$ exit vertices, there must also be $k$ middle vertices in $g_i$, and $rt_c$ must be the $k$-tuple $[1,2,\ldots,k]$.
        (That is, when $rt_c$ is considered as a map on indices of exit vertices of $g_{i-1}$, $rt_c$ is the identity map.)
      \item
        \label{Inv:2}
        If $c$ is the $B$-call edge whose source is middle vertex $j+1$ of $g_i$ and whose target is $g_{i-1}$, then $rt_c$ must meet two conditions:
        \begin{enumerate}
          \item
            \label{Inv:2a}
            It must map the exit vertices of $g_{i-1}$ one-to-one (but not necessarily onto) the exit vertices of $g_i$.
            (That is, there are no repetitions in $rt_c$.)
          \item
            \label{Inv:2b}
            It must ``compactly extend'' the set of exit vertices in $g_i$ defined by the return tuples for the previous $j$ $B$-connections:
            Let $rt_{c_1}$, $rt_{c_2}$, $\ldots$, $rt_{c_j}$ be the return tuples for the first $j$ $B$-connection edges out of $g_i$.
            Let $S$ be the set of indices of exit vertices of $g_i$ that occur in return tuples $rt_{c_1}$, $rt_{c_2}$, $\ldots$, $rt_{c_j}$, and let $n$ be the largest value in $S$.
            (That is, $n$ is the index of the rightmost exit vertex of $g_i$ that is a target of any of the return tuples $rt_{c_1}$, $rt_{c_2}$, $\ldots$, $rt_{c_j}$.)
            If $S$ is empty, then let $n$ be $0$.
    
            \hspace*{1.5ex}
            Now consider $rt_c$ ($= rt_{c_{j+1}}$).
            Let $R$ be the (not necessarily contiguous) sub-sequence of $rt_c$ whose values are strictly greater than $n$.
            Let $m$ be the size of $R$.
            Then $R$ must be exactly the sequence $[n+1, n+2, \ldots, n+m]$.
        \end{enumerate}
      \item
        \label{Inv:3}
        While a proto-CFLOBDD may be used as a substructure more than once (i.e., a proto-CFLOBDD may be {\em pointed to\/} multiple times), a proto-CFLOBDD never contains two separate {\em instances\/} of equal proto-CFLOBDDs.\footnote{
          \label{Footnote:CFLOBDDEquality}
          Equality on proto-CFLOBDDs is defined inductively on their hierarchical structure in the obvious manner.
          Two CFLOBDDs are equal when (i) their proto-CFLOBDDs are equal, and (ii) their value tuples are equal.
          When we wish to consider the possibility that \emph{multiple} data-structure instances exist that are equal, we say that such structures are ``isomorphic'' or ``equal (up to isomorphism).''  
        }
      \item
        \label{Inv:4}
        For every pair of $B$-call edges $c$ and $c'$ of grouping $g_i$, with associated return tuples $rt_c$ and $rt_{c'}$, if $c$ and $c'$ lead to level $i-1$ proto-CFLOBDDs, say $p_{i-1}$ and $p'_{i-1}$, such that $p_{i-1} = p'_{i-1}$, then the associated return tuples must be different (i.e., $rt_c \neq rt_{c'}$).
    \end{enumerate}

    A \emph{CFLOBDD} at level $k$ is a mock-CFLOBDD at level $k$ for which
    \begin{enumerate}[resume]
      \item
        \label{Inv:5}
        The grouping at level $k$ heads a proto-CFLOBDD.
      \item
        \label{Inv:6}
        The value tuple associated with the grouping at level $k$ maps each exit vertex to a \emph{distinct} value.
    \end{enumerate}
\end{definition}

\Cref{Fi:StructuralInvariantsIllustrated} illustrates structural invariants~\ref{Inv:1}, \ref{Inv:2a}, \ref{Inv:2b}, \ref{Inv:3}, \ref{Inv:4}, and~\ref{Inv:6}.
In each case, a mock-proto-CFLOBDD that violates one of the structural invariants is shown on the left, and an equivalent proto-CFLOBDD that satisfies the structural invariants is shown on the right.

%% file: 10_appendix-proofLemCorr.tex
\section{The proof of Lemma \ref{Lem:Corr}}
\label{Appendix: ProofLemCorr}

\begin{SLem}{Corr}
    Let $g$ be a level-\textit{l} grouping in $C$. For a node $n$ in $B$, if $n \triangleright g$, then BDDpatch$(u, 2^l)$ is equivalent to the proto-CFLOBDD headed by $g$. That is, there exists a bijection between the leaf nodes of BDDpatch$(u, 2^l)$ and the exit vertices of $g$ such that we always get to a corresponding leaf node/exit vertex pair for any assignment of $2^l$ variables.
\end{SLem}

\begin{proof}
    We fix $k$ and do an induction on the level-$l$.

    \paragraph{Base Case:} When $l=0$, a level-$l$ grouping is either a fork-grouping or a don't-care-grouping. In the sense of isomorphism, BDDpatch($u, 1$) can take two forms, a BDD with one leaf node, or a BDD with two leaf nodes. To strengthen the correspondence, we will refer to the former as a ``don't-care-BDD'' and the latter as a ``fork-BDD.''

    \begin{figure}
        \centering
        \includegraphics[width=1\linewidth]{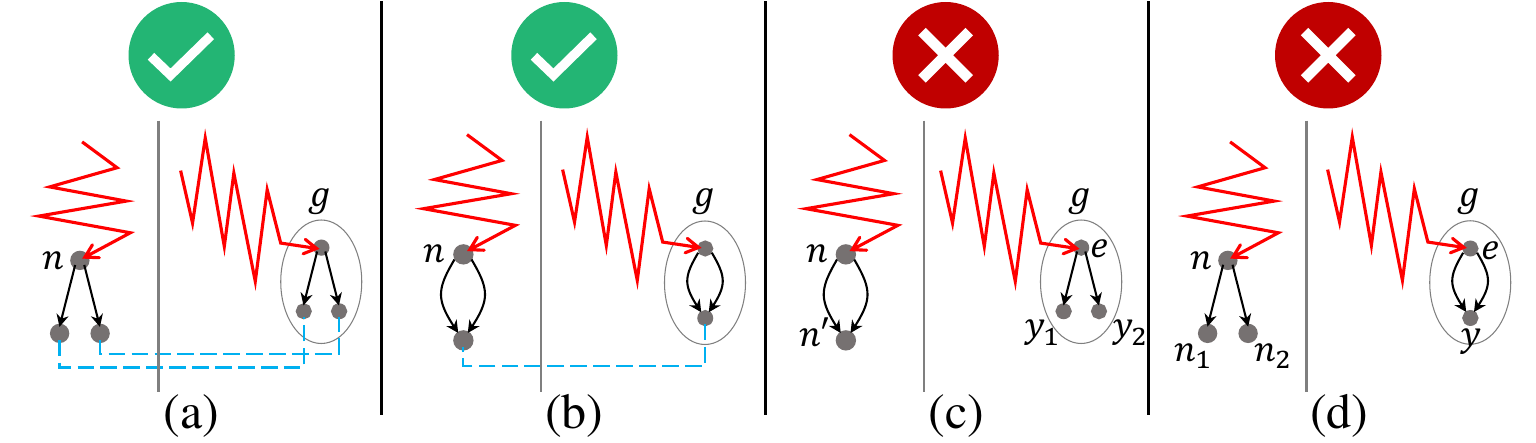}
        \caption{
\revision{
        (Color online.)
}
        The four cases for level-\textit{0} considered in the proof of \Cref{Lem:Corr}. (a) and (b) show the two cases for level-$0$ that are possible, and the blue line denotes the bijection that establishes the equivalence.
        (c) and (d) show the two impossible cases.
        }
        \label{fig:level0}
    \end{figure}

    With the two groupings and two forms of BDDs, there are---\emph{a priori}---four possibilities, as shown in \Cref{fig:level0}.
    We can easily find the bijection between leaf nodes and exit vertices for cases (a) and (b), as depicted by the blue lines. We now wish to show that if $n \triangleright g$, cases (c) and (d) are impossible.
    \begin{itemize}
        \item For (c), BDDpatch$(n, 1)$ is a don't-care-BDD and $g$ is a fork-grouping. Let $\alpha = [x_1 \mapsto a_1, \cdots, x_j \mapsto a_j]$ be the assignment that leads to $n$ and the entry vertex of $g$ (depicted as red lines). We extend it to $\alpha_0 = [x_1 \mapsto a_1, \cdots, x_j \mapsto a_j, x_{j+1} \mapsto 0]$ and $\alpha_1 = [x_1 \mapsto a_1, \cdots, x_j \mapsto a_j, x_{j+1} \mapsto 1]$.
        If we interpret $\alpha_0$ and $\alpha_1$ in $B$, we reach the same node $n'$.
        Now suppose that $\alpha_0$ and $\alpha_1$ are both extended with the same (arbitrarily chosen) assignments to variables $x_{j+2}, \cdots, x_{2^k-1}$;
        call these $\widehat{\alpha_0}$ and $\widehat{\alpha_1}$, respectively.
        Continuing from $n'$, both $\widehat{\alpha_0}$ and $\widehat{\alpha_1}$ lead to the same terminal value.
        In $C$, however, we get to different exit vertices $y_1$ and $y_2$ (immediately) after we read $x_{j+1}$.
        In this situation, there must exist some $\widehat{\alpha_0}/\widehat{\alpha_1}$ pair for which the interpretation of $x_{j+2}, \cdots, x_{2^k-1}$ via $\widehat{\alpha_0}$ (and $\widehat{\alpha_1}$) would lead to different terminal values.\footnote{\label{footnote: proofLemCorr}
          The fact that there can be two different outcomes can be demonstrated by constructing $\hat{\alpha}_0$ (and $\hat{\alpha}_1$) in an ``operational'' way.
          As long as we have not reached the terminal values, carry out the following steps:
          \begin{itemize}
            \item
              If we are at two exit vertices of the same grouping (with the same ``context''), just follow the return-edge and go to either (1) two different middle vertices, or (2) two different exit vertices of the same grouping. We would not get to the same destination because structural invariants tell us that the return map should be injective.
            \item 
              If we are at two middle vertices of the same level-$p$ grouping (with the same ``context''), we treat the B-connection as a whole and choose values for the next $2^{p-1}$ variables (for both $\hat{\alpha}_0$ and $\hat{\alpha}_1$) to reach two different exit vertices. We are assured that such an assignment exists, otherwise the two middle vertices would have the same callee and the same return-map, which violates the structural invariants.
        \end{itemize}
        Each of these steps advances through CFLOBDD $C$ along the same kind of path traversed when interpreting $C$ with respect to an assignment.
        (In the language of CFL-reachability \cite[\S5]{DBLP:conf/pods/Yannakakis90}, such a path is an 
        \textit{unbalanced-right} path
        \cite[\S4.2]{DBLP:journals/infsof/Reps98}---i.e., a suffix of a \textit{matched} path through $C$ \cite[\S3.1]{TOPLAS:SCR24}.)
        From Eqn.\ \eqref{Eq:MatchedPathLengthRecurrence}, we know that the length of each matched path in a CFLOBDD is finite.
        Consequently, the process of traversing the suffix of a matched path, as described above,
        will terminate at the outermost grouping of $C$, at which point we have two different terminal values.
        }
        This situation would contradict the assumption that $B$ and $C$ represent the same pseudo-Boolean function.
        Consequently, case (c) cannot arise.
        \item
        For (d), BDDpatch$(n, 1)$ is a fork-BDD and $g$ is a don't-care-grouping.
        We define $\alpha_0$, $\alpha_1$, $\widehat{\alpha_0}$ and $\widehat{\alpha_1}$ in a similar way to case (c).
        Whatever values are chosen for $x_{j+2}, \cdots, x_{2^k-1}$ in $\widehat{\alpha_0}$ and $\widehat{\alpha_1}$, we get to the same terminal value in $C$ (because the path taken is identical).
        In contrast, there is a way to choose values for $x_{j+2}, \cdots, x_{2^k-1}$
        (creating a particular $\widehat{\alpha_0}$/$\widehat{\alpha_1}$ pair)
        that lead to different terminal values in $B$: $n_1$ and $n_2$ are different nodes in $B$, and hence represent different residual functions.
        Again, this situation contradicts the assumption that $B$ and $C$ represent the same pseudo-Boolean function, and so case (d) cannot arise.
    \end{itemize}

    \paragraph{Induction Step:} Assume that the lemma holds for $l-1$;
    we must prove it true for $l$. 

    Let $g_A$ be $g$'s A-callee.
    We can move from $g$'s entry vertex to $g_A$'s entry vertex without interpreting any Boolean variable, and thus $n \triangleright g_A$ also holds.
    According to the inductive hypothesis, the proto-CFLOBDD headed by $g_A$ is equivalent to BDDpatch($n, 2^{l-1}$). 
    This ``equivalence'' gives us an bijection between the exit vertices of $g_A$ and the nodes at half height in BDDpatch($n, 2^l$). 
    Let $h_1, h_2, \cdots$ be the BDD nodes at half height that correspond to the ``returnees'' of $g_A$'s first, second, $\cdots$ exit vertices, respectively; let $m_1, m_2, \cdots$ be the middle vertices of $g$; let $g_{B_1}, g_{B_2}, \cdots$ be the B-callees from the middle vertices of $g$ that corresponds to $m_1, m_2, \cdots$, respectively, as shown in \Cref{fig:induction}.
    
    We will move from the $i^{\textit{th}}$ exit vertex of $g_A$ to middle vertex $m_i$ of $g$
    and then to a B-callee $g_{B_i}$ without interpreting any variables;
    thus, we have $m_1 \triangleright g_{B_1}, m_2 \triangleright g_{B_2}, \cdots$.
    By applying the inductive hypothesis on $g$'s $i^{\textit{th}}$ B-callee, we can
    identify a bijection $\mu_i$ from the exit vertices of $g_{B_i}$ to the leaves of BDDpatch($m_i, 2^{l-1})$ (which are a subset of the leaves of BDDpatch($n, 2^l$)).

    The rest of the proof is about establishing a bijection $\mu_g$ from the exit vertices of $g$ and the leaves of BDDpatch($n$, $2^l$). Intuitively, $\mu_g$ should be defined by ``inheriting'' $\mu_1, \mu_2, \cdots$: \[ \mu_g(n_x) = \mu_i(y), y\text{ is an exit vertex of } g_{B_i}\text{ and }y \text{ returns to }n_x \] 
    With such a definition, it is obvious that we will get the corresponding result for any assignment. But to make the definition a well-defined bijection,
    we need to show that the following property holds:
    let $y$ be an exit vertex of $g_{B_i}$, let $z$ be an exit vertex of $g_{B_j}$, then:
    
    \begin{itemize}
      \item
        If $\mu_i(y) \neq \mu_j(z)$, then $y$ and $z$ must return to different exit vertices of grouping $g$.
      \item
        If $\mu_i(y) = \mu_j(z)$, then $y$ and $z$ must return to the same exit vertex of $g$.
    \end{itemize}

    \begin{figure}
        \centering
        \includegraphics[width=0.4\linewidth]{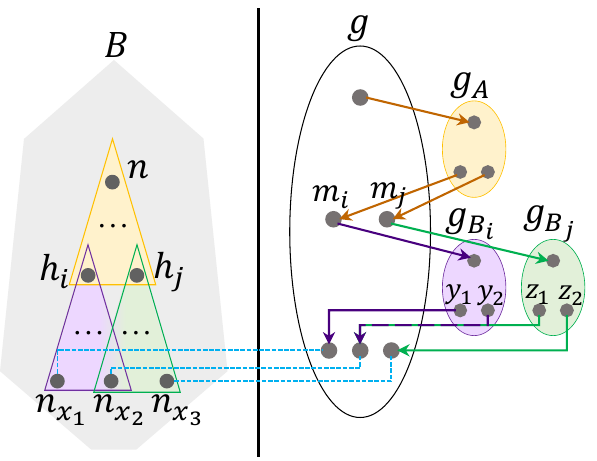}
        \caption{
\revision{
        (Color online.)
}
        A picture for the induction step in the proof of \Cref{Lem:Corr}. }
        \label{fig:induction}
    \end{figure}
    
    \Cref{fig:induction} gives an example of such $y$ and $z$. $\mu_i(y_2) = \mu_j(z_1)$, so $y_2$ and $z_1$ return to the same exit vertex of $g$; $\mu_i(y_1) \neq \mu_j(z_2)$, so $y_1$ and $z_2$ return to different exit vertices of $g$.

    For any pair of such $y$ and $z$, there are four cases to consider---from two alternatives for two conditions:
    (i) $\mu_i(y)=\mu_j(z)$ versus $\mu_i(y) \neq \mu_j(z)$, and
    (ii) $y$ and $z$ return to the same versus different exit vertices of $g$.
    We find ourselves in a situation similar to the four cases considered in the base case of the induction---we need to prove that the following two of the cases are impossible:
    \begin{enumerate}
        \item $\mu_i(y) = \mu_j(z)$, but $y$ and $z$ return to different exit vertices of $g$.
        \item $\mu_i(y) \neq \mu_j(z)$, but $y$ and $z$ return to the same exit vertex of $g$.
    \end{enumerate}
    
    Let $\beta_0$ and $\beta_1$ be two assignments of length $2^l$. Suppose that we get to $y$ and $z$, respectively, after interpreting all the variables in $\beta_0$ and $\beta_1$ in the proto-CFLOBDD headed by $g$.
    According to the definitions of $\mu_i$ and $\mu_j$, the condition $\mu_i(y) = \mu_j(z)$ indicates whether we get to the same leaf in BDDpatch($n$, $2^l$) after interpreting $\beta_0$ and $\beta_1$. Therefore, we can follow a method similar to the argument that established that cases (c) and (d) of the base case are impossible. 
    We just need to modify the definitions of $\alpha_0$ and $\alpha_1$ to be lengthened by values for $2^l$ variables
    (that is, $\beta_0$ and $\beta_1$), instead of just one variable, and then extended to $\widehat{\alpha_0}$ and $\widehat{\alpha_1}$, as in base cases (c) and (d).
    We will not repeat it here.
\end{proof}